\documentclass[aip,jmp,preprint,a4paper]{revtex4-1}
\usepackage{cancel}
\usepackage{amsmath}
\usepackage{amssymb}
\usepackage{amsthm}
\usepackage{graphicx}
\usepackage{dcolumn}
\usepackage{bm}
\usepackage{upgreek}
\usepackage{esint}

\usepackage{commands}

\theoremstyle{definition}
\newtheorem{definition}{Definition}

\theoremstyle{remark}
\newtheorem{remark}{Remark}

\theoremstyle{theorem}
\newtheorem{theorem}{Theorem}

\theoremstyle{lemma}
\newtheorem{lemma}{Lemma}

\theoremstyle{corollary}

\theoremstyle{proposition}
\newtheorem{proposition}{Proposition}

\draft 

\begin{document}


\title{Guiding center dynamics as motion on a formal slow manifold in loop space} 



\author{J. W. Burby}
\affiliation{Los Alamos National Laboratory, Los Alamos, New Mexico 87545, USA}


\date{\today}

\begin{abstract}
Since the late 1950's, the dynamics of a charged particle's ``guiding center" in a strong, inhomogeneous magnetic field have been understood in terms of near-identity coordinate transformations. The basic idea has been to approximately transform away the coupling between the fast gyration around magnetic fields lines and the remaining slow dynamics. This basic understanding now serves as a foundation for describing the kinetic theory of strongly magnetized plasmas. I present a new way to understand guiding center dynamics that does not involve complicated coordinate transformations. Starting from a dynamical systems formulation of the motion of parameterized loops in a charged particle's phase space, I identify a formal slow manifold in loop space. Dynamics on this formal slow manifold are equivalent to guiding center dynamics to all orders in perturbation theory. After demonstrating that loop space dynamics comprises an infinite-dimensional noncanonical Hamiltonian system, I recover the well-known Hamiltonian formulation of guiding center motion by restricting the (pre-) symplectic structure on loop space to the finite-dimensional guiding center formal slow manifold. 

%
\end{abstract}


\maketitle 

\section{Introduction}
Charged particles move in helical trajectories that wind around magnetic field lines. When the strength of the magnetic field is high this spinning motion is exceedingly fast and the corresponding helix is tightly wound.  Thus on timescales large compared with the cyclotron period a charged particle's trajectory resembles the motion of an approximately circular ring that may drift along or across the magnetic lines of force. Modeling these averaged, or ``guiding center," dynamics efficiently, especially in non-uniform magnetic fields, is known as the guiding center problem.

While researchers have developed various ingenious strategies\cite{kruskal58,Gardner_59,Bogoliubov_1961,Northrop_1963,Littlejohn_1981,Littlejohn_1983,Hazeltine_Waelbroeck_2004} for solving the guiding center problem, the technique that has become most widely adopted was developed by Kruskal in Ref.\,\onlinecite{Kruskal_1962} in the context of a broad class of oscillatory dynamical systems. After introducing a special sequence of near-identity coordinate transformations, Kruskal observed that short-timescale oscillations in these systems approximately decouple from the slower drift dynamics. By successively refining the near-identity transformation, the decoupling becomes increasingly complete. When applied to charged particle dynamics, Kruskal's method describes precisely the slow evolution of the fiducial ring swept out by a particle's gyration.

Kruskal's technique owes its popularity to its rigorous mathematical foundation, its tractability at low orders in perturbation theory, and its ability to explain the general phenomenon of adiabatic invariance. However, the technique also possesses several important drawbacks related to its use of complicated near-identity coordinate transformations. Three of these are:
\begin{enumerate}
\item While the method is arbitrarily accurate in principle, finding high-order approximations to the averaged dynamics requires a heroic amount of algebra, even by computer algebra standards.\cite{Burby_gc_2013}
\item While Kruskal shows that adiabatic invariance may be understood as a consequence of Noether's theorem, the symmetry that gives rise to the adiabatic invariant in Kruskal's theory is hidden. Uncovering the symmetry is akin to unearthing an infinite fossil; fractional progress requires painstaking effort, and completing the task is impossible.
\item The coordinates introduced by the method break spatial locality. Therefore the region in phase space occupied by an object that may interrupt a particle's motion, e.g. a wall, is rendered extremely complicated. 
\end{enumerate}

In this article I will describe an alternative rigorous solution of the guiding center problem that is completely free of drawbacks (2) and (3), and that suffers from drawback (1) less severely. In particular, I will realize guiding center dynamics as the restriction of loop space dynamics to a formal slow manifold.\cite{Fenichel_1979,Lorenz_1986,Lorenz_1987,Lorenz_1992} R. S. MacKay\cite{MacKay_2004} refers to such a demonstration as constructing a ``slow manifold with internal oscillation." Here loops are periodic parameterized curves in a charged particle's phase space that evolve by being dragged along by Newton's second law. 

Using the fact that loop space dynamics and its associated formal slow manifold are invariant under an obvious relabeling symmetry, I will show that the symmetry underlying adiabatic invariance is rendered obvious from the loop space perspective; it is no longer hidden. By showing that the formal slow manifold may be calculated without introducing spatially-nonlocal coordinates in phase space, I will demonstrate that the guiding center problem may be solved without breaking spatial locality. Finally, by formulating the guiding center problem as a slow manifold reduction problem, the challenge of capturing high-order effects in the drift dynamics will be recast as a the challenge of selecting initial conditions for loop space dynamics sufficiently close to the formal slow manifold. This reformulation is interesting in light of the numerical method formulated by Gear, Kaper, Kevrekidis, and Zagaris\cite{Gear_2006,Zagaris_2009} for generating points in phase space arbitrarily close to a given slow manifold.

\section{Loop space dynamics\label{sec:abstract_loops}}
This section will define and describe loop space dynamics in the context of general dynamical systems. A special feature of loop space dynamics associated with Hamiltonian systems will also be explained. Loop space dynamics governed by the Lorentz force Law will be presented as an example of the general theory. 

\subsection{Abstract Background}

Recall that a \emph{dynamical system} on a set $P$ is a one-parameter family of mappings $\mathcal{F}_t:P\rightarrow P$ with the properties $\mathcal{F}_0 = \text{id}_{P}$ and $\mathcal{F}_{t+s}= \mathcal{F}_t\circ \mathcal{F}_s$. Such a family of mappings is also referred to as a \emph{flow} or \emph{flow map}. Given a point $z\in P$, the \emph{trajectory} through $z$ is defined in terms of the flow as the parameterized curve $\mathbb{R}\rightarrow P:t\mapsto \mathcal{F}_t(z)$. When $P$ is a manifold and $\mathcal{F}_t$ is smooth, the dynamical system may be recovered from its \emph{infinitesimal generator}, which is the vector field $X$ on $P$ given by
\begin{align}
X(z) = \frac{d}{d\epsilon}\bigg|_0 \mathcal{F}_{\epsilon}(z).
\end{align}
In other words, $X(z)$ is the initial velocity of the trajectory through $z$. In many situations it is easier to specify a dynamical system by giving its infinitesimal generator instead of its flow. Nevertheless, it is sometimes more convenient to specify the flow map directly.

Given a dynamical system $\mathcal{F}_t$ on $P$, it is possible to construct various other \emph{induced} dynamical systems on spaces constructed out of $P$. In particular, there is a dynamical system induced on the space $\ell P$  of maps $\tilde{z}_0:S^1\rightarrow P$, i.e. the space of parameterized \emph{loops} in $P$. (Here the circle $S^1$ is defined as the set $\mathbb{R}\text{ mod }2\pi$.)
 The flow map $\widetilde{\mathcal{F}}_t:\ell P\rightarrow\ell P$ is given by 
\begin{align}
(\widetilde{\mathcal{F}}_t(\tilde{z}_0))(\theta) = \mathcal{F}_t(\tilde{z}_0(\theta)).
\end{align}
It will be convenient to refer to this dynamical system on $\ell P$ as the \emph{loop-parallelized dynamics} in $P$. When $P$ is a smooth manifold and $\mathcal{F}_t$ is smooth, the infinitesimal generator $\widetilde{X}_0$ of loop-parallelized dynamics is given by
\begin{align}
(\widetilde{X}_0(\tilde{z}_0))(\theta) = X(\tilde{z}_0(\theta)).
\end{align}
Here we have identified the tangent space to $\ell P$ at $\tilde{z}_0$ with the space of smooth vector fields along $\tilde{z}_0$. Thus, loop-parallelized dynamics merely parallelizes the original dynamics on $P$ over the loop parameter $\theta\in S^1$.

From here on, suppose that $P$ is a smooth manifold and that $\mathcal{F}_t$ is smooth. Starting from loop-parallelized dynamics in $P$, \emph{loop space dynamics} in $P$ is constructed as follows. Fix a smooth functional $\Omega: \ell P \rightarrow \mathbb{R}$ that is invariant under the \emph{phase shift} $\tilde{z}_0\mapsto \tilde{z}_0^\psi$ for each $\psi\in S^1$, where 
\begin{align}
\tilde{z}_0^\psi(\theta) = \tilde{z}_0(\theta+\psi).
\end{align} 
(Note that the phase shift for fixed $\psi$ defines an invertible mapping on loop space $\ell P$.) Lift loop-parallelized dynamics to the dynamical system on 
on $\ell P\times S^1$ whose infinitesimal generator is given by
\begin{align}
\widetilde{X}_0^{\Omega}(\tilde{z}_0,S) = \widetilde{X}_0(\tilde{z}_0) + \Omega(\tilde{z}_0)\partial_S.
\end{align}
This dynamical system is a lift in the sense that dynamics in first factor $\ell P$ reproduce loop-parallelized dynamics; in other words there is a one-way coupling between $\tilde{z}_0$ and $S$.
Next apply the invertible transformation $\ell P\times S^1\rightarrow \ell P\times S^1$ given by
\begin{align}
\Phi:(\tilde{z}_0,S)\mapsto (\tilde{z}_0^{-S},S).
\end{align}
This transformation may be thought of as ``spinning" the loop $\tilde{z}_0$ by the phase $S$. Loop space dynamics is then defined as the dynamical system on $\ell P\times S^1$ with infinitesimal generator $\widetilde{X}^\Omega = \Phi_* \widetilde{X}_0^{\Omega}$. Here $\Phi_*$ denotes the pushforward along $\Phi$.  Therefore $\widetilde{X}^\Omega$ is merely $\widetilde{X}_0^\Omega $ expressed in the new ``coordinates" on $\ell P\times S^1$ defined by $\Phi$. An explicit expression for the infinitesimal generator of loop space dynamics is given by
\begin{align}
\widetilde{X}^\Omega(\tilde{z},S)& = (\Phi_* \widetilde{X}_0^{\Omega})(\tilde{z},S)\nonumber\\ 
&= T\Phi\circ \widetilde{X}_0^{\Omega}\circ \Phi^{-1}(\tilde{z},S)\nonumber\\
& = T\Phi\left(\widetilde{X}_0^{\Omega}(\tilde{z}^S) + \Omega(\tilde{z}^S)\partial_S\right)\nonumber\\
& = \widetilde{X}_0^{\Omega}(\tilde{z}) -\Omega(\tilde{z})\partial_\theta \tilde{z}  + \Omega(\tilde{z})\partial_S,
\end{align}
where we have used the invariance of $\Omega$ under phase shifts. Thus, the trajectory of an element $(\tilde{z},S)\in \ell P\times S^1$ satisfies the system of equations
\begin{gather}
\partial_t\tilde{z}(\theta,t) + \Omega(\tilde{z}(t))\partial_\theta \tilde{z}(\theta,t) = X(\tilde{z}(\theta,t))\\
\dot{S}(t) = \Omega(\tilde{z}(t)).
\end{gather}

The discussion so far has emphasized the generality and geometric origins of loop space dynamics. It is also useful to be aware of the relationship between loop space dynamics and the so-called nonlinear WKB approximation.\cite{Miura_1974} Suppose $z(t)$ is a solution of the ordinary differential equation $\dot{z} = X(z)$ comprising a rapid oscillation superimposed on top of a slowly evolving envelope. The nonlinear WKB approach to describing such a solution is to leverage the ansatz $z(t) = \tilde{z}(t,S(t))$, where $S(t)$ is a rapidly rotating phase and the profile $\tilde{z}$ is periodic in its second argument. Apparently solutions of this form must satisfy $\partial_t\tilde{z}(t,S(t))+\dot{S}\partial_\theta\tilde{z}(t,S(t)) = X(\tilde{z}(t,S(t)))$. If $\dot{S}$ is chosen to approximate the rate of phase oscillations, say $\Omega$, then the scale-separation assumption implies that the profile approximately satisfies $\partial_t\tilde{z}(t,\theta)+\Omega\partial_\theta\tilde{z}(t,\theta) = X(\tilde{z}(t,\theta))$ \emph{for each} $\theta\in S^1$. This is of course the governing equation of loop space dynamics. It is interesting to notice that, from the WKB perspective, loop space dynamics would appear to represent an approximation based on scale separation. However, the geometric picture emphasized in this section shows that there is no need invoke approximations in order to provide loop space dynamics with a useful interpretation. Namely, loop space dynamics describes the evolution of parameterized loops, ``spun" by a phase, entrained in the flow of a given dynamical system. A field-theoretic generalization of this construction is discussed in Ref.\,\onlinecite{Burby_Ruiz_2019}.

Loop space dynamics associated with a Hamiltonian system enjoys an important special property that will be exploited later when connecting loop space dynamics with guiding center theory. Let $(P,-\mathbf{d}\vartheta)$ be an exact symplectic manifold (not necessarily a cotangent bundle) and fix a function $H:P\rightarrow \mathbb{R}$ that will serve as the Hamiltonian. There is then a unique vector field $X_H$ on $P$ that satisfies
\begin{align}
\iota_{X_H}\mathbf{d}\vartheta = -\mathbf{d}H.
\end{align} 
The vector field $X_H$, which is known as the Hamiltonian vector field, is the infinitesimal generator of a Hamiltonian dynamical system on $P$. 
 A useful way of characterizing the dynamical system generated by $X_H$ is in terms of the so-called phase space variational principle. This variational principle asserts that a parameterized curve $[t_1,t_2]\rightarrow P:t\mapsto z(t)$ is a trajectory of some point $z\in P$ under the dynamical system generated by $X_H$ if and only if the first fixed-endpoint variation of the action functional 
\begin{align}
A(z) = \int_{t_1}^{t_2}\bigg(\vartheta(z(t))[\dot{z}(t)] - H(z(t))\bigg)\,dt
\end{align}
vanishes at $z$. The special feature of loop space dynamics induced by $X_H$ is that they also obey a phase space variational principle. In particular, if $t\mapsto (\tilde{z}(t),S(t))$ is a trajectory of the loop space dynamics associated with $X_H$, then the first fixed-endpoint variation of the action
\begin{align}\label{loop_action_gen}
\tilde{A}(\tilde{z},S) = \int_{t_1}^{t_2} \fint \bigg(\vartheta(\tilde{z}(\theta,t))[\partial_t\tilde{z}(\theta,t)+\dot{S}(t)\partial_\theta\tilde{z}(\theta,t)] - H(\tilde{z}(\theta,t))\bigg)\,d\theta\,dt
\end{align}
vanishes at $(\tilde{z},S)$. Here $\fint = (2\pi)^{-1}\int_0^{2\pi}$. Note that this variational principle renders loop space dynamics associated with a Hamiltonian system as a classical field theory on a $(1+1)$-dimensional spacetime. Therefore Noether's theorem may be used to extract conservation laws from symmetries. The conserved quantity associated with time translation invariance is the loop energy
\begin{align}
\mathcal{H}(\tilde{z}) = \fint H(\tilde{z}(\theta))\,d\theta,\label{loop_energy_def}
\end{align}
while the conserved quantity associated with (time-independent) shifts $S\mapsto S+\psi$ is the (normalized) loop action
\begin{align}
J(\tilde{z}) = \fint \vartheta(\tilde{z}(\theta))[\partial_\theta\tilde{z}(\theta)]\,d\theta = \frac{1}{2\pi}\int_{\tilde{z}}\vartheta.\label{loop_action_def}
\end{align}

\subsection{Example: The Lorentz Force}
When the electric field is zero and the magnetic field is time-independent, Lorentz force dynamics are governed by the ordinary differential equation on $P = Q\times\mathbb{R}^3$ given by
\begin{align}
\dot{\bm{v}} =& \frac{1}{\epsilon}\bm{v}\times\bm{B}(\bm{x})\\
\dot{\bm{x}}=& \bm{v},
\end{align}
where $\bm{x}$ is contained in $Q =  \mathbb{R}^3\text{ or }Q = (S^1)^3$, $\bm{v}\in\mathbb{R}^3$, $\bm{B}$ is vector field on $Q$ that may be written as the curl of another vector field $\bm{A}$, and $\epsilon$ is the mass-to-charge ratio. (The mass-to-charge ratio may be negative.) This ordinary differential equation may be identified with the vector field $X$ on $P$ given by
\begin{align}
X(\bm{x},\bm{v}) = \bm{v}\cdot\partial_{\bm{x}} + \frac{1}{\epsilon}\bm{v}\times\bm{B}(\bm{x})\cdot\partial_{\bm{v}}.
\end{align}
After equipping $P$ with the exact symplectic form $-\mathbf{d}\vartheta$, where
\begin{align}
\vartheta(\bm{x},\bm{v}) = \frac{1}{\epsilon} \bm{A}(\bm{x})\cdot d\bm{x} +\bm{v}\cdot d\bm{x},
\end{align}
and introducing the Hamiltonian 
\begin{align}
H(\bm{x},\bm{v}) = \frac{1}{2} |\bm{v}|^2,
\end{align}
it is straightforward to verify that $X = X_H$ is the Hamiltonian vector field associated with $H$.

Loop space dynamics associated with the Lorentz force is governed by the system of equations
\begin{align}
\partial_t\tilde{\bm{v}}(\theta,t) +  \Omega(\tilde{\bm{x}}(t),\tilde{\bm{v}}(t)) \partial_\theta\tilde{\bm{v}}(\theta,t) &= \frac{1}{\epsilon}\tilde{\bm{v}}(\theta,t)\times\bm{B}(\tilde{\bm{x}}(\theta,t))\label{lorentz_loop_v}\\
\partial_t\tilde{\bm{x}}(\theta,t)  + \Omega(\tilde{\bm{x}}(t),\tilde{\bm{v}}(t)) \partial_\theta\tilde{\bm{x}}(\theta,t) & = \tilde{\bm{v}}(\theta,t)\label{lorentz_loop_x}\\
\dot{S}(t) =  \Omega(\tilde{\bm{x}}(t),\tilde{\bm{v}}(t)),\label{lorentz_loop_S}
\end{align}
where the frequency functional $\Omega$ can, in principle, be any phase-shift invariant functional on $\ell P$. The choice for $\Omega$ that will be used in this article is
\begin{align}\label{good_omega}
\Omega(\tilde{\bm{x}},\tilde{\bm{v}}) = \frac{1}{\epsilon}|\bm{B}(\overline{\bm{x}})|,
\end{align}
where $\overline{\bm{x}} = \fint \tilde{\bm{x}}\,d\theta$, which clearly satisfies the phase-shift invariance property.

The action \eqref{loop_action_gen} that governs loop space dynamics for the Lorentz force is given by
\begin{align}
\tilde{A}(\tilde{\bm{x}},\tilde{\bm{v}},S) =& \int_{t_1}^{t_2}\fint\bigg(\left[\frac{1}{\epsilon}\bm{A}(\tilde{\bm{x}}(\theta,t))+\tilde{\bm{v}}(\theta,t)\right]\cdot\partial_t\tilde{\bm{x}}(\theta,t) - \frac{1}{2}|\tilde{\bm{v}}(\theta,t)|^2\bigg)\,d\theta\,dt\nonumber\\
&+\int_{t_1}^{t_2}\dot{S}(t)\,\bigg(\fint\left[\frac{1}{\epsilon}\bm{A}(\tilde{\bm{x}}(\theta,t))+\tilde{\bm{v}}(\theta,t)\right]\cdot \partial_\theta\tilde{\bm{x}}(\theta,t)\,d\theta\bigg)\,dt.\label{ll_action}
\end{align}
The first fixed-endpoint variation of the action is given by
\begin{align}
\delta\tilde{A} = &-\int_{t_1}^{t_2} \fint \bigg(\frac{1}{\epsilon}\bm{B}(\tilde{\bm{x}})\times [\partial_t\tilde{\bm{x}}+\dot{S}\,\partial_\theta\tilde{\bm{x}}] + [\partial_t\tilde{\bm{v}}+\dot{S}\,\partial_\theta\tilde{\bm{v}}]\bigg)\cdot \delta\tilde{\bm{x}}(\theta,t)\,d\theta\,dt\nonumber\\
&-\int_{t_1}^{t_2} \fint \bigg(\tilde{\bm{v}} - [\partial_t\tilde{\bm{x}}+\dot{S}\,\partial_\theta\tilde{\bm{x}}]\bigg)\cdot\delta\tilde{\bm{v}}(\theta,t)\,d\theta\,dt\nonumber\\
& - \int_{t_1}^{t_2} \bigg(\frac{d}{dt}\fint \bigg[\frac{1}{\epsilon}\bm{A}(\tilde{\bm{x}})+\tilde{\bm{v}}\bigg]\cdot \partial_\theta\tilde{\bm{x}}\,d\theta\bigg)\delta S(t)\,dt.
\end{align}
The Euler-Lagrange equations are therefore
\begin{gather}
\frac{1}{\epsilon}\bm{B}(\tilde{\bm{x}})\times [\partial_t\tilde{\bm{x}}+\dot{S}\,\partial_\theta\tilde{\bm{x}}] + [\partial_t\tilde{\bm{v}}+\dot{S}\,\partial_\theta\tilde{\bm{v}}] = 0\label{lfl_first}\\
\tilde{\bm{v}} - [\partial_t\tilde{\bm{x}}+\dot{S}\,\partial_\theta\tilde{\bm{x}}] =0\label{lfl_second}\\
\frac{d}{dt}\fint \bigg[\frac{1}{\epsilon}\bm{A}(\tilde{\bm{x}})+\tilde{\bm{v}}\bigg]\cdot \partial_\theta\tilde{\bm{x}}\,d\theta = 0.\label{lfl_third}
\end{gather}
Note that the first two equations, Eqs.\,\eqref{lfl_first} and \eqref{lfl_second}, imply the third, Eq.\,\eqref{lfl_third}, because
\begin{align}
\frac{d}{dt}\fint \bigg[\frac{1}{\epsilon}\bm{A}(\tilde{\bm{x}})+\tilde{\bm{v}}\bigg]\cdot \partial_\theta\tilde{\bm{x}}\,d\theta = & \fint \bigg(\frac{1}{\epsilon}\bm{B}\times \partial_t\tilde{\bm{x}} + \partial_t\tilde{\bm{v}}\bigg)\cdot\partial_\theta\tilde{\bm{x}}\,d\theta - \fint \partial_t\tilde{\bm{x}}\cdot\partial_\theta\tilde{\bm{v}}\,d\theta\nonumber\\
=&-\dot{S}\fint \bigg(\frac{1}{\epsilon}\bm{B}\times \partial_\theta\tilde{\bm{x}} + \partial_\theta\tilde{\bm{v}}\bigg)\cdot\partial_\theta\tilde{\bm{x}}\,d\theta - \fint (\tilde{\bm{v}}-\dot{S}\partial_\theta\tilde{\bm{x}})\cdot \partial_\theta\tilde{\bm{v}}\nonumber\\
=& - \frac{1}{2} \fint \partial_\theta |\tilde{\bm{v}}|^2\,d\theta\nonumber\\
=& 0.
\end{align}
It follows that the Euler-Lagrange equations will be satisfied if and only if
\begin{align}
\partial_t\tilde{\bm{v}}(\theta,t) +  \dot{S}(t) \partial_\theta\tilde{\bm{v}}(\theta,t) &= \frac{1}{\epsilon}\tilde{\bm{v}}(\theta,t)\times\bm{B}(\tilde{\bm{x}}(\theta,t))\\
\partial_t\tilde{\bm{x}}(\theta,t)  +\dot{S}(t) \partial_\theta\tilde{\bm{x}}(\theta,t) & = \tilde{\bm{v}}(\theta,t).
\end{align}
Note in particular that the Euler-Lagrange equations are satisfied if $(\tilde{\bm{x}},\tilde{\bm{v}},S)$ obeys loop space dynamics, \emph{regardless of the frequency functional} $\Omega$. In other words, the initial value problem for the Euler-Lagrange equations is ill-posed. This ill-posedness, along with the redundancy of the Euler-Lagrange equation \eqref{lfl_third}, is a hallmark of \emph{gauge symmetry}. There is not necessarily an issue with the fact that the Euler-Lagrange equations are ill-posed. As long as there is some well-posed differential equation whose solutions satisfy the Euler-Lagrange equations, many of the nice tools offered by variational principles are still applicable. In this case, loop space dynamics with a given frequency functional furnish such a differential equation.

\section{A formal slow manifold in loop space\label{sec:lorentz_sm}}
\subsection{Motivating Ideas}

An old and intuitive picture of the dynamics of charged particles in a strong magnetic field replaces the particle with a charged, superconducting ring of current. One reason for introducing loop space dynamics in the study of charged particle motion is to make this intuitive picture mathematically precise. 

There is an apparent gap between the intuitive picture of moving rigid rings and the loop space description, which involves deformable loops. This is not merely a technical annoyance. The evolution of an \emph{arbitrary} loop in the Lorentz force phase space will not approximate the motion of a rigid ring in any sense. Indeed, most loops become extremely contorted as time evolves, especially when $\bm{B}$ is chosen such that Lorentz force dynamics is chaotic.

The way to establish a link between loop space dynamics and the rigid ring picture is to introduce the concept of a slow manifold. Roughly speaking, a slow manifold is a special submanifold in the phase space of a dynamical system with multiple timescales. When an initial condition is chosen to lie on the slow manifold, its subsequent time evolution will remain close to the slow manifold for a long period of time. Thus, a slow manifold is an example of an \emph{almost invariant set}. Moreover, motion on the slow manifold only weakly couples to the fast timescale. This is the sense in which is a slow manifold is ``slow."

The remainder of this Section will demonstrate that the phase space for loop space dynamics contains a formal slow manifold. Moreover, points on this formal slow manifold may be identified with rigid rings in phase space. Interestingly, these rigid rings are not geometric circles. Instead their shape is described by a set of non-trivial \emph{shape functions} whose asymptotic expansion in powers of $\epsilon$ may be computed systematically. In light of the intuitive picture of guiding center dynamics as dynamics of rigid charged superconducting rings, these results strongly suggest that loop motion on the formal slow manifold corresponds in some way to guiding center dynamics. This intuition will be justified in Section \ref{gc_proof} by proving that dynamics on the formal slow manifold is equivalent to guiding center dynamics to all orders in perturbation theory.

The purpose of Section \ref{ham_struc} will be to demonstrate the sense in which the rigid rings that support guiding center motion are ``superconducting." According to Eq.\,\eqref{lfl_third}, the dynamics of a general loop conserves action. In particular, the dynamics of a loop that evolves on the formal slow manifold conserves action. At leading-order, it will turn out that the expression for the action of a rigid loop is the magnetic flux through the loop, i.e. the usual magnetic moment adiabatic invariant is recovered. Thus, to leading order, the rigid rings conserve flux exactly as superconductors do. At higher orders, this picture has to be distorted slightly because the exact expression for the action of a rigid ring differs from the flux. As a more complete way of describing the all-orders picture, and in order to illuminate the simple origins of the symmetry underlying adiabatic invariance for charged particles, I will show how the Hamiltonian formulation of guiding center dynamics due to Littlejohn\cite{Littlejohn_1981} may be recovered by restricting the presymplectic form on loop space associated with the action functional \eqref{ll_action} to the formal slow manifold. 
\subsection{Fast-Slow Systems and Their Slow Manifolds}
A useful class of dynamical systems for the precise study of slow manifolds consists of the \emph{fast-slow systems}.
\begin{definition}
A \emph{fast-slow dynamical system} is a dynamical system on a cartesian product $P=X\times Y$ of Banach spaces $X,Y$ whose infinitesimal generator $(\dot{x},\dot{y})$ has the form
\begin{align}
\epsilon\,\dot{y} =& f_\epsilon(x,y) \\
\dot{x} = & g_\epsilon(x,y),
\end{align}
where $f_\epsilon,g_\epsilon $ depend smoothly on $\epsilon$ in some open interval containing $\epsilon = 0$ and $D_y f_0(x,y)$ is invertible whenever $f_0(x,y) = 0$.
\end{definition}
\begin{remark}
The technical hypothesis on $D_yf_0$ ensures that the limiting differential algebraic equation (DAE),
\begin{align}
 0 & = f_0(x,y)\label{slave_limit}\\
 \dot{x} & = g_0(x,y) ,\label{slow_evo_limit}
\end{align}
has differentiation index $1$. In other words, when solving the DAE, Eq.\,\eqref{slave_limit} may first be eliminated by solving for $y$ as a function of $x$, giving a function $y_0^*(x)$. The function $y_0^*(x)$ may then be substituted into Eq.\,\eqref{slow_evo_limit} in order to obtain an autonomous ordinary differential equation for $x$. Fast-slow systems therefore provide a paradigm for studying dimensionality reduction.
\end{remark}

In the fast-slow setting, a \emph{formal slow manifold} may be defined precisely as follows.

\begin{definition}
Given a fast-slow system, $\epsilon \dot{y} = f_\epsilon(x,y)$, $\dot{x} = g_\epsilon(x,y)$, a \emph{formal slow manifold} is a formal power series
\begin{align}\label{slaving_function_gen}
y_\epsilon^*(x) = y_0^*(x) + \epsilon \,y_1^*(x) + \epsilon^2\,y_2^*(x) + \dots
\end{align}
that satisfies \emph{the invariance equation}
\begin{align}\label{invariance_equation}
\epsilon Dy_\epsilon^*(x)[g_\epsilon(x,y_\epsilon^*(x))] = f_\epsilon(x,y_\epsilon^*(x)),
\end{align}
to all orders in $\epsilon$.
\end{definition}

\begin{remark}
Note that if $y_\epsilon^*$ is a genuine solution of the invariance equation for each $\epsilon$, then the set $\Gamma_\epsilon = \{(x,y)\mid y = y_\epsilon^*(x)\} $ is invariant under the dynamics of the fast-slow system for each $\epsilon$. Therefore, when $y_\epsilon^*$ is a formal slow manifold, it is suggestive, though not rigorous, to think of the ``graph" of $y_\epsilon^*$ as an invariant manifold for each $\epsilon$. This is the rationale behind referring to an asymptotic series as a ``manifold." 
\end{remark}

One of the main motivations for studying fast-slow systems is the fact that such systems always contain unique formal slow manifolds, as described by the following Proposition.

\begin{proposition}\label{SM_existence}
Associated with each fast-slow system is a unique formal slow manifold. Moreover, the coefficient $y_k^*$ of the formal slow manifold may be computed algorithmically for any $k\in\{0,1,2,\dots\}$. In particular, the first two coefficients are determined by the equations
\begin{align}
f_0(x,y_0^*(x)) &= 0\label{shape_zero}\\
y_1^*(x) & = [D_yf_0(x,y_0^*(x))]^{-1}\bigg[ Dy_0^*(x)[g_0(x,y_0^*(x))] - f_1(x,y_0^*(x))\bigg],\label{shape_one}
\end{align}
where $f_1 = \frac{d}{d\epsilon}\big|_0 f_\epsilon$.
\end{proposition}

Thus, fast-slow systems always contain formal invariant sets given as graphs of the fast variable $y$ over the slow variable $x$. Dynamics on such a set are formally prescribed by the infinitesimal generator $\dot{x} = g_\epsilon(x,y_\epsilon^*(x))$ on $X$. Because $g_\epsilon$ depends smoothly on $\epsilon$ in a neighborhood of $0$, dynamics on the formal slow manifold apparently do not involve the $O(\epsilon)$ timescale, and are in this sense slow.

Of course, there is no reason to expect that in general the series defining a formal slow manifold converges to give a true invariant set on which dynamics is slow. In the \emph{normally hyperbolic} case in finite dimensions, where the eigenvalues of $D_yf_0(x,y_0^*(x))$ are purely real, Fenichel\cite{Fenichel_1979} effectively established the convergence of the series using transversality arguments. In the \emph{normally elliptic} case, where the eigenvalues of $D_yf_0(x,y_0^*(x))$ are purely imaginary, the series are known to diverge in general due to resonance between the fast normal dynamics and the slow dynamics. However, for sufficiently smooth fast-slow systems, truncations of the series $y_\epsilon^*$ often define \emph{almost invariant sets}, meaning trajectories that begin near the truncated slow manifold remain nearby for long periods of time.\cite{MacKay_2004} In the analytic case, the time for ``sticking" to a truncated slow manifold may even be exponentially long. This state of affairs might be summarized by saying the series defining a formal slow manifold is meaningful even when it does not converge. This point is amplified vividly by Vanneste in Ref.\,\onlinecite{Vanneste_2008}, who applies Borel summation to a normally-elliptic slow manifold to define an ``optimal" almost invariant set. What emerges from this analysis is a detailed picture of exponentially-small high-frequency oscillations that are generated spontaneously by motions along the optimally-invariant set.

In the following two subsections, \ref{sub_sec:split} and \ref{slow_man_sub}, I will show that loop space dynamics associated with the Lorentz force may be written as a fast-slow system, and then explicitly compute the first two coefficients in the series $y_\epsilon^*$ defining the associated formal slow manifold. These low-order terms in the series will strongly suggest that dynamics on the formal slow manifold in loop space correspond in some way to guiding center dynamics. I will then prove the equivalence of dynamics on the formal slow manifold with guiding center dynamics to all orders in perturbation theory in subsection \ref{gc_proof}. In so doing, I will have shown that guiding center dynamics constitutes an example of what MacKay\cite{MacKay_2004} calls a ``a slow manifold with internal oscillation."

\subsection{Fast-Slow formulation of Lorentz Loop Dynamics\label{sub_sec:split}}

Consider now the problem of determining wether loop space dynamics associated with the Lorentz force comprise a fast-slow system. The infinitesimal generator for these dynamics is given in Eqs.\,\eqref{lorentz_loop_v}-\eqref{lorentz_loop_S}. The first step in finding a fast-slow split for this dynamical system is to introduce a decomposition of $\ell P$ into mean and fluctuating subspaces, $\ell P = P\oplus \widetilde{\ell P}$, where
\begin{align}
\widetilde{\ell P} = \bigg\{(\widehat{\mathbb{X}},\widehat{\mathbb{V}})\in\ell P\bigg| \fint  \widehat{\mathbb{X}}(\theta)\,d\theta = \fint \widehat{\mathbb{V}}(\theta)\,d\theta = 0\bigg\}.
\end{align}
For the remainder of this article, elements in $P$ will be denoted $(\overline{\bm{x}},\overline{\bm{v}})$. Thus, the phase space for  loop space dynamics is now expressed as $P\times \widetilde{\ell P}\times S^1$, with typical elements denoted $(\overline{\bm{x}},\overline{\bm{v}},\widehat{\mathbb{X}},\widehat{\mathbb{V}},S)$. The relationship between $(\overline{\bm{x}},\overline{\bm{v}},\widehat{\mathbb{X}},\widehat{\mathbb{V}},S)$ and the original loop space variables $(\tilde{\bm{x}},\tilde{\bm{v}},S)$ is
\begin{align}
\tilde{\bm{x}} =& \overline{\bm{x}} + \widehat{\mathbb{X}}\\
\tilde{\bm{v}} =& \overline{\bm{v}} + \widehat{\mathbb{V}}. 
\end{align}

The second step is to scale the fluctuating particle position $\widehat{\mathbb{X}}$ according to $\widehat{\mathbb{X}}\mapsto \epsilon^{-1} \widehat{\mathbb{X}}$. The transformed variable will be denoted $\widehat{\bm{\rho}} = \epsilon^{-1} \widehat{\mathbb{X}}$. The complete expression for the transformation applied in this second step is then $(\overline{\bm{x}},\overline{\bm{v}},\widehat{\mathbb{X}},\widehat{\mathbb{V}},S)\mapsto (\overline{\bm{x}},\overline{\bm{v}},\widehat{\bm{\rho}},\widehat{\mathbb{V}},S)$, which may be regarded as an invertible mapping from $P\times \widetilde{\ell P}\times S^1$ into itself because $\widetilde{\ell P}$ is scale invariant. The relationship between the original loop space variables $(\tilde{\bm{x}},\tilde{\bm{v}},S)$ and $(\overline{\bm{x}},\overline{\bm{v}},\widehat{\bm{\rho}},\widehat{\mathbb{V}},S)$ is given by
\begin{align}
\tilde{\bm{x}} =& \overline{\bm{x}} + \epsilon \widehat{\bm{\rho}}\\
\tilde{\bm{v}} =& \overline{\bm{v}} + \widehat{\mathbb{V}}. 
\end{align}

The third step is to parameterize the mean velocity variable $\overline{\bm{v}}$ as follows. Let $\bm{e}_1,\bm{e}_2$ be orthogonal unit vector fields on $\mathbb{R}^3$ that are everywhere orthogonal to the vector field $\bm{B}$. (As usual when discussing strongly magnetized particles, the field $\bm{B}$ is assumed to be nowhere vanishing.) Set $\bm{b} = \bm{B}/|\bm{B}|$ and assume $\bm{b} = \bm{e}_1\times \bm{e}_2$. Now introduce the mapping $(\overline{\bm{x}},\overline{\bm{v}})\mapsto (\overline{\bm{x}},\overline{u},\overline{v}_1,\overline{v}_2)$, where 
\begin{align}
\overline{u} =& \bm{b}(\overline{\bm{x}})\cdot\overline{\bm{v}}\\
\overline{v}_1 = &\bm{e}_1(\overline{\bm{x}})\cdot \overline{\bm{v}}\\
\overline{v}_2 = &\bm{e}_2(\overline{\bm{x}})\cdot \overline{\bm{v}}.
\end{align}
This mapping amounts to expressing $\overline{\bm{v}}$ in a moving orthonormal frame aligned with the magnetic field. The change of variables on discrete loop space in the third step is then given by $(\overline{\bm{x}},\overline{\bm{v}},\widehat{\bm{\rho}},\widehat{\mathbb{V}},S)\mapsto (\overline{\bm{x}},\overline{u},\overline{v}_1,\overline{v}_2,\widehat{\bm{\rho}},\widehat{\mathbb{V}},S)$. The relationship between the original loop space variables $(\tilde{\bm{x}},\tilde{\bm{v}},S)$ and $ (\overline{\bm{x}},\overline{u},\overline{v}_1,\overline{v}_2,\widehat{\bm{\rho}},\widehat{\mathbb{V}},S)$ is given by
\begin{align}
\tilde{\bm{x}} =& \overline{\bm{x}} + \epsilon \widehat{\bm{\rho}}\\
\tilde{\bm{v}} = & \overline{u}\bm{b}(\overline{\bm{x}}) + \overline{v}_1\bm{e}_1(\overline{\bm{x}}) + \overline{v}_2\bm{e}_2(\overline{\bm{x}}) + \widehat{\mathbb{V}}.
\end{align}
We remind those readers familiar with conventional guiding center theory that $(\widehat{\bm{\rho}},\widehat{\mathbb{V}})$ is an arbitrary element of $\widetilde{\ell P}$. In particular, $\widehat{\bm{\rho}}$ is \emph{not} required to be orthogonal to $\bm{b}$. 

The fourth and final step is to parameterize the fluctuating velocity variable $\widehat{\mathbb{V}}$ as follows. First decompose $\widehat{\mathbb{V}}$ into the sum of its first Fourier harmonic $\widehat{\mathbb{V}}_1 $ and its higher harmonic content $\widehat{\mathbb{V}}_{2+}$,
\begin{align}
\widehat{\mathbb{V}} = \widehat{\mathbb{V}}_1 + \widehat{\mathbb{V}}_{2+}.
\end{align}
Then parameterize the first harmonic $\widehat{\mathbb{V}}_1$ using the vectors $\widehat{\mathbb{V}}_1^+,\widehat{\mathbb{V}}_1^-$ according to
\begin{align}
\widehat{\mathbb{V}}_1(\theta) = \widehat{\mathbb{V}}_1^+ \cos\theta + \widehat{\mathbb{V}}_1^- \sin\theta.
\end{align}
Next express $\widehat{\mathbb{V}}_1^+,\widehat{\mathbb{V}}_1^-$ in the moving frame $(\bm{b},\bm{e}_1,\bm{e}_2)$ as
\begin{align}
\widehat{\mathbb{V}}_1^+ =& {u}^+ \bm{b}(\overline{\bm{x}}) + {v}_1^+\bm{e}_1(\overline{\bm{x}}) + {v}_2^+ \bm{e}_2(\overline{\bm{x}}) \\
\widehat{\mathbb{V}}_1^- =& {u}^- \bm{b}(\overline{\bm{x}}) + {v}_1^-\bm{e}_1(\overline{\bm{x}}) + {v}_2^- \bm{e}_2(\overline{\bm{x}}) .
\end{align}
Finally, introduce the components of the \emph{adiabatic velocity}
\begin{align}
{w}_1 = & \frac{1}{2} ({v}_1^+ - {v}_2^-)\\
{w}_2 = & \frac{1}{2} ({v}_2^+ + {v}_1^-),
\end{align}
and the components of the \emph{non-adiabatic velocity}
\begin{align}
{\omega}_1 = &\frac{1}{2} ({v}_1^+ + {v}_2^-) \\
{\omega}_2 = & \frac{1}{2} ({v}_2^+ - {v}_1^-).
\end{align}
This sequence of definitions may be interpreted as a mapping $ (\overline{\bm{x}},\overline{u},\overline{v}_1,\overline{v}_2,\widehat{\bm{\rho}},\widehat{\mathbb{V}},S)\mapsto (\overline{\bm{x}},\overline{u},\overline{v}_1,\overline{v}_2,\widehat{\bm{\rho}},{u}^+,{u}^-,{w}_1,{w}_2,{\omega}_1,{\omega}_2,\widehat{\mathbb{V}}_{2+},S) $. The relationship between the original loop space dynamics phase space variables $(\tilde{\bm{x}},\tilde{\bm{v}},S)$ and the new variables is given explicitly by
\begin{align}
\tilde{\bm{x}} = & \overline{\bm{x}} + \epsilon\widehat{\bm{\rho}}\\
\tilde{\bm{v}} = & (\overline{u}+u^+\cos\theta + u^-\sin\theta)\bm{b}(\overline{\bm{x}})+  \overline{\bm{v}}_\perp\nonumber\\
& +   (\cos\theta \mathbb{I} + \sin\theta \bm{b}_\times)\cdot\bm{\omega}_\perp + (\cos\theta\mathbb{I} - \sin\theta \bm{b}_\times)\cdot \bm{w}_\perp + \widehat{\mathbb{V}}_{2+},
\end{align}
where the following useful shorthand notation has been introduced:
\begin{align}
\overline{\bm{v}}_\perp = & \overline{v}_1\bm{e}_1(\overline{\bm{x}}) + \overline{v}_2 \bm{e}_2(\overline{\bm{x}})\\
\bm{w}_\perp = & w_1 \bm{e}_1(\overline{\bm{x}}) + w_2 \bm{e}_2(\overline{\bm{x}})= \frac{1}{2}(1-\bm{b}\bm{b})\cdot(\widehat{\mathbb{V}}_1^+ + \bm{b}_\times\cdot \widehat{\mathbb{V}}_1^-)\\
\bm{\omega}_\perp = & \omega_1 \bm{e}_1(\overline{\bm{x}}) + \omega_2 \bm{e}_2(\overline{\bm{x}})=\frac{1}{2}(1-\bm{b}\bm{b})\cdot (\widehat{\mathbb{V}}_1^+ - \bm{b}_\times\cdot \widehat{\mathbb{V}}_1^-),
\end{align}
and the tensor $\bm{b}_\times$ is defined by $\bm{b}_\times \cdot \bm{a} = \bm{b}\times \bm{a}$.

The infinitesimal generator for discrete loop space dynamics expressed in terms of the final set of new variables and the scaled phase $\mathcal{S} = \epsilon S$ is given by
\begin{align}
\epsilon\, \partial_t u^+ = & - |\bm{B}|(\overline{\bm{x}}) u^- + \epsilon\,\bigg[ \overline{\bm{v}}\cdot \nabla\bm{b}\cdot(\bm{w}_\perp + \bm{\omega}_\perp) + 2\bm{b}\cdot \fint \cos\theta\,\tilde{\bm{v}}\times\delta\bm{B}\,d\theta  \bigg]\label{uplus_evo}\\
\epsilon\, \partial_t u^- = & |\bm{B}|(\overline{\bm{x}}) u^+ + \epsilon\,\bigg[\overline{\bm{v}}\cdot\nabla\bm{b}\cdot (\bm{w}_\perp\times\bm{b}-\bm{\omega}_\perp\times\bm{b}) + 2\bm{b}\cdot\fint \sin\theta\,\tilde{\bm{v}}\times\delta\bm{B}\,d\theta\bigg] \label{uminus_evo}\\
\epsilon\,\partial_t\omega_1 = & 2|\bm{B}|(\overline{\bm{x}})\,\omega_2 + \epsilon \fint (\cos\theta \bm{e}_1+\sin\theta\bm{e}_2)\cdot\tilde{\bm{V}}\times\delta\bm{B}\,d\theta\nonumber\\
&\hspace{7em}-\epsilon\bigg[\frac{1}{2}\overline{\bm{v}}\cdot\nabla\bm{b}\cdot(u^+\bm{e}_1 + u^-\bm{e}_2) - \overline{\bm{v}}\cdot\bm{R}\,\omega_2\bigg] \label{omega1_evo}\\
\epsilon\,\partial_t\omega_2 = & -2|\bm{B}|(\overline{\bm{x}})\,\omega_1 + \epsilon \fint (\cos\theta\bm{e}_2 - \sin\theta\bm{e}_1)\cdot\tilde{\bm{v}}\times\delta\bm{B}\,d\theta\nonumber\\
&\hspace{7em}-\epsilon\bigg[\frac{1}{2}\overline{\bm{v}}\cdot\nabla\bm{b}\cdot(u^+\bm{e}_2 - u^- \bm{e}_1) + \overline{\bm{v}}\cdot\bm{R}\,\omega_1\bigg]\label{omega2_evo}\\
\epsilon\,\partial_t \overline{v}_1 = & |\bm{B}|(\overline{\bm{x}}) \,\overline{v}_2 + \epsilon \bm{e}_1\cdot\fint \tilde{\bm{v}}\times\delta\bm{B}\,d\theta - \epsilon\bigg[ \overline{u}\,\overline{\bm{v}}\cdot\nabla\bm{b}\cdot\bm{e}_1 - \overline{\bm{v}}\cdot\bm{R}\,\overline{v}_2\bigg]   \label{v1bar_evo}\\
\epsilon\,\partial_t\overline{v}_2 = & -|\bm{B}|(\overline{\bm{x}}) \,\overline{v}_1 + \epsilon \bm{e}_2\cdot\fint \tilde{\bm{v}}\times\delta\bm{B}\,d\theta - \epsilon\bigg[ \overline{u}\,\overline{\bm{v}}\cdot\nabla\bm{b}\cdot\bm{e}_2 + \overline{\bm{v}}\cdot\bm{R}\,\overline{v}_1\bigg] \label{v2bar_evo}\\
\epsilon\, \partial_t \widehat{\mathbb{V}}_{2+} = & \widehat{\mathbb{V}}_{2+}\times\bm{B}(\overline{\bm{x}}) - |\bm{B}|(\overline{\bm{x}})\,\partial_\theta \widehat{\mathbb{V}}_{2+} + \epsilon\, \pi_{2+}\left(\tilde{\bm{v}}\times \delta\bm{B}\right)\label{harmv_evo}\\
\epsilon\, \partial_t\widehat{\bm{\rho}} = & \widehat{\mathbb{V}}- |\bm{B}|(\overline{\bm{x}})\,\partial_\theta\widehat{\bm{\rho}}\label{rho_evo}\\
\partial_t w_1 = &  -\frac{1}{2} \overline{\bm{v}}\cdot\nabla\bm{b}\cdot (u^+\bm{e}_1 - u^-\bm{e}_2) + \overline{\bm{v}}\cdot\bm{R}\, w_2 + \fint (\cos\theta \bm{e}_1 - \sin\theta \bm{e}_2)\cdot\tilde{\bm{v}}\times\delta\bm{B}\,d\theta \label{w1_evo}\\
\partial_t w_2 = &  - \frac{1}{2}\overline{\bm{v}}\cdot\nabla\bm{b}\cdot(u^+ \bm{e}_2 + u^- \bm{e}_1) - \overline{\bm{v}}\cdot\bm{R}\,w_1 + \fint (\cos\theta \bm{e}_2+\sin\theta \bm{e}_1)\cdot\tilde{\bm{v}}\times\delta\bm{B}\,d\theta \label{w2_evo}\\
\partial_t\overline{u} = & \overline{\bm{v}}\cdot\nabla\bm{b}\cdot\overline{\bm{v}} +  \bm{b}\cdot\fint \tilde{\bm{v}}\times\delta\bm{B}\,d\theta\label{ubar_evo}\\
\partial_t\overline{\bm{x}} = & \overline{u}\bm{b}(\overline{\bm{x}}) + \overline{\bm{v}}_\perp\label{xbar_evo}\\
\dot{\mathcal{S}} = & |\bm{B}(\overline{\bm{x}})|
\end{align}
where $\pi_{2+}$ is the $L^2$-orthogonal projection onto the space of Fourier harmonics greater than or equal to $2$, $\bm{R} = (\nabla\bm{e}_1)\cdot\bm{e}_2$, and the symbol $\delta\bm{B}$ is defined as
\begin{align}
\delta\bm{B} = & \int_0^1 \widehat{\bm{\rho}}\cdot\nabla\bm{B}(\overline{\bm{x}}+ \lambda\epsilon \widehat{\bm{\rho}})\,d\lambda.
\end{align}
Note that $\bm{B}(\tilde{\bm{x}}) = \bm{B}(\overline{\bm{x}}) + \epsilon\, \delta\bm{B}$.

\begin{theorem}\label{fast_slow_proof}
When written in terms of the dependent variables \[x=(\overline{\bm{x}},\overline{u},w_1,w_2,\mathcal{S})\] and \[y = (\widehat{\bm{\rho}},\overline{v}_1,\overline{v}_2,u^+,u^-,\omega_1,\omega_2,\widehat{\mathbb{V}}_{2+}),\] loop space dynamics associated with the Lorentz force is a fast-slow dynamical system. 
\end{theorem}

\begin{proof}
By Eqs.\,\eqref{uplus_evo}-\eqref{rho_evo}, the evolution equation for $y$ is $\epsilon\, \dot{y} = f_\epsilon(x,y)$, where $f_\epsilon $ depends smoothly on $\epsilon$ and $f_0$ is given by 
\begin{align}
f_0(x,y) &= (\dot{\widehat{\bm{\rho}}}_0,(\dot{\overline{v}}_1)_0,(\dot{\overline{v}}_1)_0,\dot{u}^+_0,\dot{u}^-_0,(\dot{\omega}_1)_0,(\dot{\omega}_2)_0,(\dot{\widehat{\mathbb{V}}}_{2+})_0)\\
\dot{\widehat{\bm{\rho}}}_0 & =  (u^+\cos\theta + u^-\sin\theta)\bm{b}(\overline{\bm{x}})+   (\cos\theta \mathbb{I} + \sin\theta \bm{b}_\times)\cdot\bm{\omega}_\perp \nonumber\\
&\quad+ (\cos\theta\mathbb{I} - \sin\theta \bm{b}_\times)\cdot \bm{w}_\perp + \widehat{\mathbb{V}}_{2+}     - |\bm{B}|(\overline{\bm{x}})\,\partial_\theta\widehat{\bm{\rho}}\\
(\dot{\overline{v}}_1)_0 & =|\bm{B}|(\overline{\bm{x}}) \,\overline{v}_2 \\
(\dot{\overline{v}}_2)_0 & = -|\bm{B}|(\overline{\bm{x}}) \,\overline{v}_1\\
\dot{u}^+_0 & =  - |\bm{B}|(\overline{\bm{x}}) u^- \\
\dot{u}^-_0 & =   |\bm{B}|(\overline{\bm{x}}) u^+\\
(\dot{\omega}_1)_0 & = 2|\bm{B}|(\overline{\bm{x}})\,\omega_2\\
(\dot{\omega}_2)_0 & = -2|\bm{B}|(\overline{\bm{x}})\,\omega_1\\
(\dot{\widehat{\mathbb{V}}}_{2+})_0 & =  \widehat{\mathbb{V}}_{2+}\times\bm{B}(\overline{\bm{x}}) - |\bm{B}|(\overline{\bm{x}})\,\partial_\theta \widehat{\mathbb{V}}_{2+} .
\end{align}
The derivative $D_yf_0(x,y)$ is therefore given by
\begin{align}
D_y f_0(x,y)[\delta y] =& (\delta\dot{\widehat{\bm{\rho}}}_0,\delta(\dot{\overline{v}}_1)_0,\delta(\dot{\overline{v}}_1)_0,\delta\dot{u}^+_0,\delta\dot{u}^-_0,\delta(\dot{\omega}_1)_0,\delta(\dot{\omega}_2)_0,\delta(\dot{\widehat{\mathbb{V}}}_{2+})_0)  \\
\delta\dot{\widehat{\bm{\rho}}}_0 & =  (\delta u^+\cos\theta + \delta u^-\sin\theta)\bm{b}(\overline{\bm{x}})+   (\cos\theta \mathbb{I} + \sin\theta \bm{b}_\times)\cdot\delta \bm{\omega}_\perp \nonumber\\
&\quad+ \delta \widehat{\mathbb{V}}_{2+}     - |\bm{B}|(\overline{\bm{x}})\,\partial_\theta\delta\widehat{\bm{\rho}}\label{delta_rho}\\
\delta(\dot{\overline{v}}_1)_0 & =|\bm{B}|(\overline{\bm{x}}) \,\delta\overline{v}_2 \label{delta_v1}\\
\delta(\dot{\overline{v}}_2)_0 & = -|\bm{B}|(\overline{\bm{x}}) \,\delta\overline{v}_1\\
\delta\dot{u}^+_0 & =  - |\bm{B}|(\overline{\bm{x}}) \delta u^- \\
\delta\dot{u}^-_0 & =   |\bm{B}|(\overline{\bm{x}}) \delta u^+\\
\delta(\dot{\omega}_1)_0 & = 2|\bm{B}|(\overline{\bm{x}})\,\delta \omega_2\\
\delta(\dot{\omega}_2)_0 & = -2|\bm{B}|(\overline{\bm{x}})\,\delta \omega_1\label{delta_omega2}\\
\delta(\dot{\widehat{\mathbb{V}}}_{2+})_0 & =  \delta \widehat{\mathbb{V}}_{2+}\times\bm{B}(\overline{\bm{x}}) - |\bm{B}|(\overline{\bm{x}})\,\partial_\theta \delta\widehat{\mathbb{V}}_{2+} .\label{delta_vharm}
\end{align}
In order to assess the invertibility of $D_y f_0(x,y)$, first note that the space $Y$ is given by
\begin{align}
Y = \widetilde{\ell Q} \times\mathbb{R}^2\times \mathbb{R}^2\times \mathbb{R}^2 \times \ell_{2+}\mathbb{R}^3,
\end{align}
where $\ell_{2+}\mathbb{R}^3$ is the set of loops in $\mathbb{R}^3$ with zero'th and first Fourier harmonics equal to $0$.
Now fix an arbitrary $y_s\in Y$ (``$s$'' stands for ``source") with components
\begin{align}
y_s = (\widehat{\bm{\rho}}_s,(\overline{v}_1)_s,(\overline{v}_2)_s, u^+_s,u^-_s,(\omega_1)_s,(\omega_2)_s,(\widehat{\mathbb{V}}_{2+})_s),
\end{align}
and consider solving the equation $D_y f_0(x,y)[\delta y] = y_s $ for $\delta y$. By Eqs.\,\eqref{delta_v1}-\eqref{delta_omega2}, $(\delta \overline{v}_1,\delta \overline{v}_2)$, $(\delta u^+,\delta u^-)$, and $(\delta\omega_1,\delta\omega_2)$ may be expressed in terms of $y_s$ by solving three separate $2\times 2$ matrix equations, giving the result
\begin{align}
\delta\overline{v}_1 = &-\frac{1}{|\bm{B}|(\overline{\bm{x}})} (\overline{v}_2)_s\label{inv_v1bar}\\
\delta\overline{v}_2 = &\frac{1}{|\bm{B}|(\overline{\bm{x}})} (\overline{v}_1)_s\label{inv_v2bar}\\
\delta u^+ = &\frac{1}{|\bm{B}|(\overline{\bm{x}})} u^-_s\label{inv_uplus}\\
\delta u^- = &-\frac{1}{|\bm{B}|(\overline{\bm{x}})} u^+_s\label{inv_uminus}\\
\delta\omega_1 = &-\frac{1}{2|\bm{B}|(\overline{\bm{x}})}(\omega_2)_s\label{inv_omega1}\\
\delta\omega_2 = &\frac{1}{2|\bm{B}|(\overline{\bm{x}})}(\omega_1)_s.\label{inv_omega2}
\end{align}
In particular, $\delta\bm{\omega}_\perp = \delta\omega_1\bm{e}_1+\delta\omega_2\bm{e}_2 = \frac{1}{2|\bm{B}|(\overline{\bm{x}})}\bm{b}\times (\bm{\omega}_\perp)_s$. By Eq.\,\eqref{delta_vharm}, the Fourier harmonics of $\delta\widehat{\mathbb{V}}_{2+}$ are determined by the sequence of equations
\begin{align}
(\widehat{\mathbb{V}}_k^+)_s = & |\bm{B}|(\overline{\bm{x}}) \delta\widehat{\mathbb{V}}^+_k\times\bm{b} - k |\bm{B}|(\overline{\bm{x}}) \delta\widehat{\mathbb{V}}_k^-\\
(\widehat{\mathbb{V}}_k^-)_s=&|\bm{B}|(\overline{\bm{x}}) \delta\widehat{\mathbb{V}}^-_k\times\bm{b} + k |\bm{B}|(\overline{\bm{x}}) \delta\widehat{\mathbb{V}}_k^+,
\end{align}
where $k$ is any integer greater than or equal to $2$. For each $k$, this linear system may be solved for $(\delta\widehat{\mathbb{V}}_k^+,\delta\widehat{\mathbb{V}}_k^-)$, giving
\begin{align}
\delta\widehat{\mathbb{V}}_{k}^+ =  & \frac{1}{k|\bm{B}|(\overline{\bm{x}})}\bigg( \bm{b}\bm{b}\cdot  (\widehat{\mathbb{V}}_k^-)_s+ \frac{k^2}{k^2-1} (\bm{b}\times (\widehat{\mathbb{V}}_k^-)_s)\times\bm{b} + \frac{k}{k^2-1}(\widehat{\mathbb{V}}_k^+)_s\times\bm{b}\bigg)\label{inv_deltav_plus}\\
\delta\widehat{\mathbb{V}}_k^- = & -\frac{1}{k|\bm{B}|(\overline{\bm{x}})}\bigg( \bm{b}\bm{b}\cdot  (\widehat{\mathbb{V}}_k^+)_s+ \frac{k^2}{k^2-1} (\bm{b}\times (\widehat{\mathbb{V}}_k^+)_s)\times\bm{b} - \frac{k}{k^2-1}(\widehat{\mathbb{V}}_k^-)_s\times\bm{b}\bigg).\label{inv_deltav_minus}
\end{align}
Finally, upon decomposing Eq.\,\eqref{delta_rho} into first- and higher-order harmonics, the first-order harmonics of $\delta\widehat{\bm{\rho}}$ may be expressed as
\begin{align}
\delta\widehat{\bm{\rho}}_1^+ = &\quad \frac{1}{|\bm{B}|(\overline{\bm{x}})}(\widehat{\bm{\rho}}_s)^-_1 + \frac{1}{|\bm{B}|^2(\overline{\bm{x}})} u_s^+\bm{b} + \frac{1}{2|\bm{B}|^2(\overline{\bm{x}})}(\bm{\omega}_\perp)_s\label{inv_rho1_plus}\\
\delta\widehat{\bm{\rho}}_1^- = & -\frac{1}{|\bm{B}|(\overline{\bm{x}})}(\widehat{\bm{\rho}}_s)^+_1 + \frac{1}{|\bm{B}|^2(\overline{\bm{x}})} u_s^-\bm{b} + \frac{1}{2|\bm{B}|^2(\overline{\bm{x}})}\bm{b}\times(\bm{\omega}_\perp)_s,\label{inv_rho1_minus}
\end{align}
and the $k$'th-order harmonics ($k\geq 2$) of $\delta\widehat{\bm{\rho}}$ may be expressed as
\begin{align}
\delta\widehat{\bm{\rho}}_k^+ = & -\frac{1}{|\bm{B}|(\overline{\bm{x}})} (\widehat{\bm{\rho}}_s)^+_k + \frac{1}{k|\bm{B}|(\overline{\bm{x}})}\delta\widehat{\mathbb{V}}_k^+\label{inv_rhok_plus}\\
\delta\widehat{\bm{\rho}}_k^- = &\quad \frac{1}{|\bm{B}|(\overline{\bm{x}})} (\widehat{\bm{\rho}}_s)^-_k - \frac{1}{k|\bm{B}|(\overline{\bm{x}})}\delta\widehat{\mathbb{V}}_k^-,\label{inv_rhok_minus}
\end{align}
with $\delta\widehat{\mathbb{V}}_k^+$ and $\delta\widehat{\mathbb{V}}_k^-$ given in Eqs.\,\eqref{inv_deltav_plus} and \eqref{inv_deltav_minus}. It is now apparent that $D_y f_0(x,y)$ is invertible for all $(x,y)$ with explicit inverse given by Eqs.\,\eqref{inv_v1bar} - \eqref{inv_omega2}, \eqref{inv_deltav_plus} - \eqref{inv_deltav_minus}, \eqref{inv_rho1_plus}-\eqref{inv_rho1_minus}, and \eqref{inv_rhok_plus}-\eqref{inv_rhok_minus}.
\end{proof}

\begin{remark}
The factor of $2$ appearing in $(\dot{\omega}_1)_0$ and $(\dot{\omega}_2)_0$ is caused by spinning the loops by the phase $S$.
\end{remark}

\subsection{Finding The Slow Manifold in Loop Space\label{slow_man_sub}}
Now that loop space dynamics associated with the Lorentz force have been identified with a fast-slow system, Proposition \ref{SM_existence} implies that there is a unique formal slow manifold in loop space given by the formal series $y_\epsilon^* = y_0^* +\epsilon y_1^* + \epsilon^2 y_2^* +\dots$. Interestingly, because Theorem \ref{fast_slow_proof} shows that the slow variable $x$ lives in a finite-dimensional space, this formal slow manifold is finite-dimensional. (The dimension is $7$.) Therefore the series $y_\epsilon^*$ may be interpreted as describing the shape of rigid loops in loop space. The term ``rigid" is appropriate in this case because the loops on the formal slow manifold are determined by only $6$ real parameters. (The loop shape is independent of $\mathcal{S}$.) The first two terms in shape function series $y_\epsilon^* = y_0^* +\epsilon y_1^* + \epsilon^2 y_2^* +\dots$ may be computed as follows.

\subsubsection{The Leading-Order Shape Function $y_0^*$}
According to Eq.\,\eqref{shape_zero} in Proposition \ref{SM_existence}, the leading-order shape functions for the slow loops in loop phase space, i.e. $y_0^*$, are given by
\begin{align}
\widehat{\bm{\rho}}^*_0(\theta) = & \sin\theta\, \frac{\bm{w}_\perp}{|\bm{B}(\overline{\bm{x}})|} - \cos\theta\,\frac{\bm{w}_\perp\times\bm{b}(\overline{\bm{x}})}{|\bm{B}(\overline{\bm{x}})|}\\
(\overline{\bm{v}}_\perp)^*_0 = & 0\\
(u^+)^*_0 = & 0\\
(u^-)^*_0 = & 0\\
(\bm{\omega}_\perp)_0^* = & 0\\
(\widehat{\mathbb{V}}_{2+})^*_0 = & 0.
\end{align}
The leading-order dynamics on the formal slow manifold, i.e. the $\epsilon\rightarrow 0$ limit of $\dot{x} = g_\epsilon(x,y_\epsilon^*(x))$, are therefore governed by
\begin{align}
\partial_t w_1 = &  \overline{u}\,\bm{b}(\overline{\bm{x}})\cdot\bm{R}\, w_2 + \frac{1}{2} w_1 \overline{u}\,\bm{b}\cdot\nabla\text{ln}|\bm{B}| + \frac{1}{2}w_2 \overline{u}\,\bm{b}\cdot\nabla\times\bm{b} \label{w1_evo_slow0}\\
\partial_t w_2 = &   - \overline{u}\,\bm{b}(\overline{\bm{x}})\cdot\bm{R}\,w_1 -\frac{1}{2}w_1 \overline{u}\,\bm{b}\cdot\nabla\times\bm{b} + \frac{1}{2}w_2\overline{u}\,\bm{b}\cdot\nabla\text{ln}|\bm{B}| \label{w2_evo_slow0}\\
\partial_t\overline{u} = &-\frac{1}{2}\frac{|\bm{w}_\perp|^2}{|\bm{B}|}\bm{b}\cdot\nabla|\bm{B}|\label{ubar_evo_slow0}\\
\partial_t\overline{\bm{x}} = & \overline{u}\bm{b}(\overline{\bm{x}}).\label{xbar_evo_slow0}\\
\dot{\mathcal{S}} = & |\bm{B}(\overline{\bm{x}})|
\end{align}
Note that these leading-order slow evolution equations have the exact conservation laws
\begin{align}
\partial_t\frac{|\bm{w}_\perp|^2}{2|\bm{B}|} = 0,
\end{align}
corresponding the conservation of action, and 
\begin{align}
\partial_t \left(\frac{1}{2}\overline{u}^2 + \frac{1}{2}|\bm{w}_\perp|^2\right) = 0,
\end{align}
corresponding to the conservation of energy. These conservation laws may be recovered as limiting forms of the exact conservation laws for action and energy for loop space dynamics. It follows that Eqs.\,\eqref{ubar_evo_slow0} and \eqref{xbar_evo_slow0} are identical to the leading-order equations describing guiding center dynamics. (See the discussion below Eq.\,(15) in Ref.\,\onlinecite{Littlejohn_bounce_1982})

\subsubsection{The First-Order Shape Function $y_1^*$}
Again referring to Proposition \ref{SM_existence}, the first-order shape functions are determined by the system of equations
\begin{gather}
0 = -|\bm{B}(\overline{\bm{x}})|\,(u^-)^*_1 + \overline{u}\bm{b}(\overline{\bm{x}})\cdot\nabla\bm{b}\cdot\bm{w}_\perp + 2\bm{b}\cdot \fint\cos\theta\,\tilde{\bm{v}}_0\times \delta\bm{B}_0\,d\theta\\
0 = |\bm{B}(\overline{\bm{x}})|\,(u^+)^*_1 + \overline{u}\,\bm{b}(\overline{\bm{x}})\cdot\nabla\bm{b}\cdot(\bm{w}_\perp\times\bm{b}) + 2\bm{b}\cdot\fint\sin\theta\,\tilde{\bm{v}}_0\times \delta\bm{B}_0\,d\theta\\
0 = 2|\bm{B}(\overline{\bm{x}})|(\omega_2)_1^* + \fint (\cos\theta\bm{e}_1 +\sin\theta\bm{e}_2)\cdot\tilde{\bm{v}}_0\times\delta\bm{B}_0\,d\theta\\
0 = - 2|\bm{B}(\overline{\bm{x}})|(\omega_1)_1^* +\fint (\cos\theta\bm{e}_2-\sin\theta\bm{e}_1)\cdot\tilde{\bm{v}}_0\times\delta\bm{B}_0\,d\theta\\
0 = |\bm{B}(\overline{\bm{x}})| (\overline{v}_2)_1^* + \bm{e}_1\cdot\fint \tilde{\bm{v}}_0\times\delta\bm{B}_0\,d\theta - \overline{u}^2\bm{b}\cdot\nabla\bm{b}\cdot\bm{e}_1\\
0 = -|\bm{B}(\overline{\bm{x}})| (\overline{v}_1)_1^* + \bm{e}_2\cdot\fint \tilde{\bm{v}}_0\times\delta\bm{B}_0\,d\theta - \overline{u}^2\bm{b}\cdot\nabla\bm{b}\cdot\bm{e}_2\\
 0  = (\widehat{\mathbb{V}}_{2+})^*_1\times\bm{B}(\overline{\bm{x}}) - |\bm{B}(\overline{\bm{x}})| \partial_\theta  (\widehat{\mathbb{V}}_{2+})^*_1 + \pi_{2+}\left(\tilde{\bm{v}}_0\times \delta\bm{B}_0\right)\\
 \frac{\overline{u}}{|\bm{B}|}\cos\theta\bigg(\frac{1}{2}\tau \bm{w}_\perp + \frac{1}{2}k_\parallel \bm{w}_\perp\times\bm{b} - \bm{w}_\perp\times\bm{\kappa}\bigg)\nonumber\\
 + \frac{\overline{u}}{|\bm{B}|}\sin\theta\bigg(-\frac{1}{2}k_\parallel \bm{w}_\perp + \frac{1}{2}\tau \bm{w}_\perp\times\bm{b} - \bm{\kappa}\cdot\bm{w}_\perp \bm{b}\bigg)\nonumber\\
  = [ (u^+)_1^*\cos\theta + (u^-)_1^*\sin\theta]\bm{b}(\overline{\bm{x}}) + (\cos\theta\mathbb{I} + \sin\theta \bm{b}_\times)\cdot (\bm{\omega}_\perp)_1^*  \nonumber\\
 + (\widehat{\mathbb{V}}_{2+})_1^* - |\bm{B}| \partial_\theta \widehat{\bm{\rho}}_1^*,
\end{gather}
where $0$ in a subscript denotes evaluation using the leading-order shape function $y_0^*$, and I have introduced the useful shorthand notation
\begin{align}
\tau& = \bm{b}\cdot\nabla\times\bm{b}\\
\bm{\kappa} &= \bm{b}\cdot\nabla\bm{b}\\
k_\parallel &  = \bm{b}\cdot\nabla \text{ln}|\bm{B}|
\end{align}

The solution of these equations may be found with the help of the identities
\begin{align}
\widetilde{\bm{F}}_L = & \tilde{\bm{v}}_0\times \delta\bm{B}_0\nonumber\\
			        = & \overline{u}\bm{b}\times (\widehat{\bm{\rho}}_0^*\cdot\nabla\bm{B}) + \widehat{\bm{\rho}}_0\widehat{\bm{\rho}}_0: \nabla\bm{B}\,\bm{B} - |\bm{B}|\widehat{\bm{\rho}}_0\widehat{\bm{\rho}}_0\cdot\nabla|\bm{B}|\\
\fint \widetilde{\bm{F}}_L\,d\theta = & - \frac{|\bm{w}_\perp|^2}{2|\bm{B}|}\,\nabla|\bm{B}|\\
\fint \cos\theta\widetilde{\bm{F}}_L\,d\theta = & -\frac{\overline{u}}{2|\bm{B}|} \bm{b}\times ([\bm{w}_\perp\times\bm{b}]\cdot\nabla\bm{B})\\
\fint\sin\theta\widetilde{\bm{F}}_L\,d\theta = & \frac{\overline{u}}{2|\bm{B}|}\bm{b}\times (\bm{w}_\perp\cdot\nabla\bm{B})\\
\pi_{2+}(\widetilde{\bm{F}}_L) = & \frac{|\bm{w}_\perp|^2}{2|\bm{B}|}\bigg([\bm{c}\bm{c} - \bm{a}\bm{a}]:(\nabla\bm{B}\,\bm{b} -\nabla|\bm{B}|\mathbb{I}) \bigg)\,\cos 2\theta\nonumber\\
   & -  \frac{|\bm{w}_\perp|^2}{2|\bm{B}|}\bigg([\bm{a}\bm{c} + \bm{c}\bm{a}]:(\nabla\bm{B}\,\bm{b} -\nabla|\bm{B}|\mathbb{I})\bigg)\,\sin 2\theta
\end{align}
where $\bm{a} = \bm{w}_\perp/|\bm{w}_\perp|$ and $\bm{c} = \bm{w}_\perp\times\bm{b}/|\bm{w}_\perp|$. Explicitly, the solution is given by
\begin{align}
(u^+)_1^* = & - \frac{\overline{u}\bm{\kappa}\cdot \bm{w}_\perp\times\bm{b}}{|\bm{B}|}\\
(u^-)_1^* = & \frac{\overline{u}\bm{\kappa}\cdot\bm{w}_\perp}{|\bm{B}|}\\
(\omega_1)^*_1 = & \frac{\overline{u}}{4|\bm{B}|^2} \bm{w}_\perp\cdot \left( \nabla\bm{B}\cdot\bm{b}_\times - \bm{b}_\times \cdot\nabla\bm{B}\right)\cdot\bm{e}_1\\
(\omega_2)^*_1 = & \frac{\overline{u}}{4|\bm{B}|^2} \bm{w}_\perp\cdot \left( \nabla\bm{B}\cdot\bm{b}_\times - \bm{b}_\times \cdot\nabla\bm{B}\right)\cdot\bm{e}_2\\
(\overline{v}_1)^*_1 = &  -\frac{(\mu_0\nabla|\bm{B}| +\overline{u}^2 \bm{\kappa})\times\bm{b}}{|
\bm{B}|}\cdot\bm{e}_1\label{v1bar_star1}\\
(\overline{v}_2)^*_1 = &  -\frac{(\mu_0\nabla|\bm{B}| + \overline{u}^2\bm{\kappa})\times\bm{b}}{|
\bm{B}|}\cdot\bm{e}_2\label{v2bar_star1}\\
(\widehat{\mathbb{V}}_2^+)^*_1 = & \frac{1}{2}\mu_0[\bm{a}\bm{c} + \bm{c}\bm{a}]:\nabla\bm{b}\,\bm{b} + \mu_0 [\bm{a}\bm{c}+\bm{c}\bm{a}]\cdot\nabla\text{ln}|\bm{B}|\\
(\widehat{\mathbb{V}}_2^-)^*_1 = & \frac{1}{2}\mu_0[\bm{c}\bm{c}-\bm{a}\bm{a}]:\nabla\bm{b}\,\bm{b} + \mu_0[\bm{c}\bm{c}-\bm{a}\bm{a}]\cdot\nabla\text{ln}|\bm{B}|\\
\widehat{\bm{\rho}}^*_1 = & -\frac{\overline{u}|\bm{w}_\perp|}{|\bm{B}|^2}\bigg(\frac{1}{4}\bm{a}\cdot(\nabla\bm{b} + \bm{b}_\times\cdot\nabla\bm{b}\cdot\bm{b}_\times) + \frac{1}{2}k_\parallel \bm{a} - \frac{1}{2}\tau\bm{c} + 2\bm{\kappa}\cdot\bm{a}\,\bm{b}\bigg)\cos\theta\nonumber\\
& +\frac{\overline{u}|\bm{w}_\perp|}{|\bm{B}|^2}\bigg(\frac{1}{4}\bm{a}\cdot(\nabla\bm{b}\cdot\bm{b}_\times -\bm{b}_\times\cdot\nabla\bm{b}) -\frac{1}{2}\tau \bm{a} - \frac{1}{2}k_\parallel \bm{c} - 2\bm{\kappa}\cdot\bm{c}\,\bm{b}\bigg)\sin\theta\nonumber\\
& -\frac{1}{4}\frac{|\bm{w}_\perp|^2}{|\bm{B}|^2} \bigg(\frac{1}{2}[\bm{c}\bm{c}-\bm{a}\bm{a}]:\nabla\bm{b}\,\bm{b} + [\bm{c}\bm{c}-\bm{a}\bm{a}]\cdot\nabla\text{ln}|\bm{B}|\bigg)\cos2\theta\nonumber\\
& +\frac{1}{4}\frac{|\bm{w}_\perp|^2}{|\bm{B}|^2}\bigg(\frac{1}{2}[\bm{a}\bm{c}+\bm{c}\bm{a}]:\nabla\bm{b}\,\bm{b} + [\bm{a}\bm{c}+\bm{c}\bm{a}]\cdot\nabla\text{ln}|\bm{B}|\bigg)\sin 2\theta,
\end{align}
where $(\widehat{\mathbb{V}}_{2+})^*_1 = (\widehat{\mathbb{V}}_2^+)^*_1\,\cos\theta + (\widehat{\mathbb{V}}_2^-)^*_1\,\sin\theta$.
Note in particular that Eqs.\,\eqref{v1bar_star1}-\eqref{v2bar_star1} lead to the following improved expression for the time derivative of $\overline{\bm{x}}$ on the formal slow manifold:
\begin{align}
\dot{\overline{\bm{x}}} = \overline{u}\,\bm{b}(\overline{\bm{x}})  - \epsilon \frac{(\mu_0\nabla|\bm{B}| + \overline{u}^2\bm{\kappa})\times\bm{b}}{|\bm{B}|}+ O(\epsilon^2).
\end{align}
The $O(\epsilon)$ correction term reproduces the famous $\nabla B$ and curvature drifts from guiding center theory.\cite{Northrop_1963} Thus, evidence is mounting that loop dynamics restricted to the formal slow manifold is closely related to guiding center dynamics. 

\subsection{Slow Loops Move as Guiding Centers\label{gc_proof}}
Apparently an explanation is required for the low-order coincidence of guiding center dynamics with loop space dynamics on the formal slow manifold. In order to show that motion on the formal slow manifold in loop space is in fact \emph{equivalent} to guiding center dynamics to all orders in perturbation theory, it is useful to draw upon Kruskal's description\cite{Kruskal_1962} of guiding center dynamics based on near-identity coordinate transformations. In Kruskal's approach, first a set of coordinates $(\zeta,\xi^1,\dots,\xi^5)$ on the $6$-dimensional phase space for a single charged particle is found with the following property: in these coordinates, the Lorentz force equation takes the form
\begin{align}
\dot{\zeta} =& \frac{1}{\epsilon}\Omega_{-1}(\bm{\xi}) + \Omega_0(\zeta,\bm{\xi})\\
\dot{\bm{\xi}} = & \bm{U}_0(\zeta,\bm{\xi}),
\end{align}
where $\bm{\xi} = (\xi^1,\dots, x^5)$ and $\Omega_{-1}$ is nowhere vanishing. Here $\zeta$ is an angular variable, and the functions $\Omega_0,\bm{U}_0$ are $2\pi$-periodic in $\zeta$. Next a sequence of near-identity coordinate transformations $\Phi^{(N)}_\epsilon:(\zeta,\bm{\xi})\mapsto (\overline{\zeta}_N,\overline{\bm{\xi}}_N)$ is found that decouples $\zeta$ from $\bm{\xi}$ with increasing accuracy. Here near-identity means $\Phi^{(N)}_0 = \text{id}$ for each $N$ and $(\Phi^{(N)}_\epsilon)^{-1}\circ \Phi^{(N+1)}_\epsilon = \text{id} + O(\epsilon^N)$. In addition, approximate decoupling means that in the new coordinates $ (\overline{\zeta}_N,\overline{\bm{\xi}}_N)$ the Lorentz force equation takes the form
\begin{align}
\dot{\overline{\zeta}}_N = & \frac{1}{\epsilon} \Omega_{-1}(\overline{\bm{\xi}}_N) + \overline{\Omega}_\epsilon^{(N)}(\overline{\bm{\xi}}_N) + \epsilon^N \delta\Omega^{(N)}_\epsilon(\overline{\zeta}_N,\overline{\bm{\xi}}_N)\label{partial_trans_one}\\
\dot{\overline{\bm{\xi}}}_N = & \overline{U}_\epsilon^{(N)}(\overline{\bm{\xi}}_N) + \epsilon^N\delta U^{(N)}_\epsilon(\overline{\zeta}_N,\overline{\bm{\xi}}_N),\label{partial_trans_two}
\end{align}
where $\delta\Omega^{(N)}_\epsilon,\delta U^{(N)}_\epsilon = O(1)$ as $\epsilon\rightarrow 0$.
Assuming the magnetic field is $C^\infty$, the integer $N$ may in principle be made as large as one would like. Therefore the difference between Eqs.\,\eqref{partial_trans_one} -\eqref{partial_trans_two} and those same equations with $\delta\Omega^{(N)}_\epsilon,\delta U^{(N)}_\epsilon = 0$ may be be made arbitrarily small. The system of equations given by dropping $\delta\Omega^{(N)}_\epsilon,\delta U^{(N)}_\epsilon$ are \emph{the guiding center equations of order} $N$. 

By considering all values of $N$, the transformations $\Phi_\epsilon^{(N)}$ define a formal transformation $\Phi_\epsilon:(\zeta,\bm{\xi})\mapsto (\underline{\zeta},\underline{\bm{\xi}})$ that decouples $\underline{\zeta}$ from $\underline{\bm{\xi}}$ to all orders in $\epsilon$. The formal power series for the transformed Lorentz force infinitesimal generator is given by
\begin{align}
\dot{\underline{\zeta}} &= \frac{1}{\epsilon} \Omega_{-1}(\underline{\bm{\xi}}) + \overline{\Omega}_\epsilon(\underline{\bm{\xi}})\label{all_orders_gc_one}\\
\dot{\underline{\bm{\xi}}} & = \overline{U}_\epsilon(\underline{\bm{\xi}}),\label{all_orders_gc_two}
\end{align}
where $\overline{\Omega}_\epsilon,\overline{U}_\epsilon$ are each formal power series in $\epsilon$. Equations \eqref{all_orders_gc_one} - \eqref{all_orders_gc_two} are \emph{the all-orders guiding center equations}. The first $N$ terms in the all-orders guiding center equations agree with the guiding center equations of order $N$ for each $N$. I will now demonstrate that the formal slow dynamics on the formal slow manifold in loop space agrees with the all-orders guiding center equations.

As a way to motivate my argument, consider first the following characterization of loop space dynamics associated with the Lorentz force in terms of the special coordinate systems used in Kruskal's theory.

\begin{lemma}\label{equivalence_lemma}
The coordinates $(\zeta,\bm{\xi})$ for the Lorentz force may be chosen so that loop space dynamics associated with the Lorentz force are equivalent to 
\begin{align}
\partial_t\widetilde{\zeta}(\theta,t) + \frac{1}{\epsilon} \omega_c(\widetilde{\zeta}(t),\widetilde{\bm{\xi}}(t))\,\partial_\theta\widetilde{\zeta}(\theta,t)= & \frac{1}{\epsilon}\Omega_{-1}(\widetilde{\bm{\xi}}(\theta,t)) + \Omega_0(\widetilde{\zeta}(\theta,t),\widetilde{\bm{\xi}}(\theta,t)) \\
\partial_t\widetilde{\bm{\xi}}(\theta,t) + \frac{1}{\epsilon}  \omega_c(\widetilde{\zeta}(t),\widetilde{\bm{\xi}}(t))\,\partial_\theta\widetilde{\bm{\xi}}(\theta,t) = & U_0(\widetilde{\zeta}(\theta,t),\widetilde{\bm{\xi}}(\theta,t))\\
\omega_c(\widetilde{\zeta},\widetilde{\bm{\xi}}) =&  \Omega_{-1}\left(\fint \widetilde{\bm{\xi}}(\theta^\prime) \,d\theta^\prime\right)\label{omegac_def}\\
\dot{S}(t) = & \frac{1}{\epsilon}\omega_c(\widetilde{\zeta}(t),\widetilde{\bm{\xi}}(t)),
\end{align}
where $\widetilde{\zeta}(\theta)$ may be chosen to be of the form $\widetilde{\zeta}(\theta) = \theta + \widetilde{\nu}(\theta)$ with a single-valued $\widetilde{\nu}$. Equivalently, the transformed loop $(\widetilde{\zeta}_N(\theta),\widetilde{\bm{\xi}}_N(\theta)) = \Phi^{(N)}_\epsilon(\widetilde{\zeta}(\theta),\widetilde{\bm{\xi}}(\theta))$ satisfies
\begin{gather}
\partial_t\widetilde{\zeta}_N(\theta,t) + \frac{1}{\epsilon}\omega_c\left((\widetilde{\Phi}^{(N)}_\epsilon)^{-1}(\widetilde{\zeta}_N(t),\widetilde{\bm{\xi}}_N(t))\right)\,\partial_\theta \widetilde{\zeta}_N(\theta,t) =\nonumber\\ \frac{1}{\epsilon} \Omega_{-1}(\widetilde{\bm{\xi}}_N(\theta,t)) + \overline{\Omega}_\epsilon^{(N)}(\widetilde{\bm{\xi}}_N(\theta,t)) + \epsilon^N \delta\Omega_\epsilon^{(N)}(\widetilde{\zeta}(\theta,t),\widetilde{\bm{\xi}}_N(\theta,t))\\
\partial_t\widetilde{\bm{\xi}}_N(\theta,t) +  \frac{1}{\epsilon}\omega_c\left((\widetilde{\Phi}^{(N)}_\epsilon)^{-1}(\widetilde{\zeta}_N(t),\widetilde{\bm{\xi}}_N(t))\right)\,\partial_\theta \widetilde{\bm{\xi}}_N(\theta,t) = \nonumber\\\overline{U}_\epsilon^{(N)}(\widetilde{\bm{\xi}}_N(\theta,t)) + \epsilon^N \delta U_\epsilon^{(N)}(\widetilde{\zeta}_N(\theta,t),\widetilde{\bm{\xi}}_N(\theta,t) )\\
\dot{S}(t) = \frac{1}{\epsilon}\omega_c\left((\widetilde{\Phi}^{(N)}_\epsilon)^{-1}(\widetilde{\zeta}_N(t),\widetilde{\bm{\xi}}_N(t))\right)
\end{gather}
where the loop $(\widetilde{\Phi}^{(N)}_\epsilon)^{-1}(\widetilde{\zeta}_N(t),\widetilde{\bm{\xi}}_N(t))(\theta) = (\Phi_\epsilon^{(N)})^{-1}(\widetilde{\zeta}_N(\theta,t),\widetilde{\bm{\xi}}_N(\theta,t))$.
\end{lemma}

This result is a simple corollary of the fact that constructing loop space dynamics commutes with applying coordinate transformations. Indeed, suppose that $\dot{z} = Y(z)$ is the infinitesimal generator of a smooth dynamical system on a manifold $M\ni z$, and let $\phi:M\rightarrow M:z\mapsto \underline{z}$ be a diffeomorphism. We may apply the diffeomorphism to $M$, thereby obtaining the transformed infinitesimal generator $\underline{Y} = \phi_*Y$, and then construct the corresponding loop space dynamics:
\begin{align}
\partial_t \widetilde{\underline{z}}(\theta,t) + \underline{\Omega}(\widetilde{\underline{z}}(t))\,\partial_\theta\widetilde{\underline{z}}(\theta,t) = \underline{Y}(\widetilde{\underline{z}}(\theta,t)).\label{trans_loop}
\end{align}
Equivalently, we may first construct loop space dynamics associated with $Y$:
\begin{align}
\partial_t\widetilde{z}(\theta,t) + \Omega(\widetilde{z}(t))\,\partial_\theta\widetilde{z}(\theta,t) = Y(\widetilde{z}(\theta,t)),
\end{align}
and then inquire as to the dynamics of the transformed loop $\widetilde{\underline{z}}(\theta) = \phi(\widetilde{z}(\theta))$. Applying the chain rule gives \eqref{trans_loop} with $\underline{\Omega}(\widetilde{\underline{z}}) = \Omega(\widetilde{\phi}^{-1}(\widetilde{\underline{z}}))$, where $\widetilde{\phi}^{-1}(\widetilde{\underline{z}})(\theta) = \phi^{-1}(\widetilde{\underline{z}}(\theta))$.

It is therefore a small step to replace $\Phi^{(N)}_\epsilon$ in Lemma \ref{equivalence_lemma} with the formal all-orders transformation $\Phi_\epsilon$, and thereby obtain the following expression for the loop space infinitesimal generator that is valid to all orders in perturbation theory.

\begin{lemma}\label{equivalence_lemma_2}
The infinitesimal generator for loop space dynamics associated with the Lorentz force is equivalent to the formal series
\begin{align}
\partial_t\widetilde{\underline{\zeta}}(\theta,t) + \frac{1}{\epsilon}\omega_c\left((\widetilde{\Phi}_\epsilon)^{-1}(\widetilde{\underline{\zeta}}(t),\widetilde{\underline{\bm{\xi}}}(t))\right)\,\partial_\theta \widetilde{\underline{\zeta}}(\theta,t) =& \frac{1}{\epsilon} \Omega_{-1}(\widetilde{\underline{\bm{\xi}}}(\theta,t)) + \overline{\Omega}_\epsilon(\widetilde{\underline{\bm{\xi}}}(\theta,t))\label{aligned_zeta} \\
\partial_t\widetilde{\underline{\bm{\xi}}}(\theta,t) +  \frac{1}{\epsilon}\omega_c\left((\widetilde{\Phi}_\epsilon)^{-1}(\widetilde{\underline{\zeta}}(t),\widetilde{\underline{\bm{\xi}}}(t))\right)\,\partial_\theta \widetilde{\underline{\bm{\xi}}}(\theta,t) = &\overline{U}_\epsilon(\widetilde{\underline{\bm{\xi}}}(\theta,t))\label{aligned_xi} \\
\dot{S}(t) =& \frac{1}{\epsilon}\omega_c\left((\widetilde{\Phi}^{(N)}_\epsilon)^{-1}(\widetilde{\underline{\zeta}}(t),\widetilde{\underline{\bm{\xi}}}(t))\right),
\end{align}
where $\widetilde{\underline{\zeta}}(\theta) = \theta + \widetilde{\underline{\nu}}(\theta)$.
\end{lemma}

This perturbative characterization of loop space dynamics is useful because (a) the coefficients of the all-orders guiding center equations appear explicitly, and (b) the fast-slow split for loop space dynamics has become especially simple.

\begin{proposition}\label{formal_fast_slow}
In terms of the formulation given in Lemma \ref{equivalence_lemma_2}, the fast and slow variables for loop space dynamics associated with the Lorentz force are given by
\begin{align}
x =& (\overline{\underline{\nu}},\overline{\underline{\bm{\xi}}},\mathcal{S})\\
y =& (\widehat{\underline{\nu}},\widehat{\bm{\varrho}}),
\end{align}
where $\mathcal{S} = \epsilon\, S$ and
\begin{align}
\overline{\underline{\nu}} =& \fint \widetilde{\underline{\nu}}(\theta^\prime)\,d\theta^\prime\\
\overline{\underline{\bm{\xi}}} = & \fint \widetilde{\bm{\xi}}(\theta^\prime)\,d\theta^\prime\\
\widehat{\underline{\nu}}(\theta) = & \widetilde{\underline{\nu}}(\theta) -  \fint \widetilde{\underline{\nu}}(\theta^\prime)\,d\theta^\prime\\
\epsilon\,\widehat{\bm{\varrho}}(\theta) = & \widetilde{\underline{\bm{\xi}}}(\theta) - \fint \widetilde{\bm{\xi}}(\theta^\prime)\,d\theta^\prime.
\end{align}
\end{proposition}

\begin{proof}
According to Eqs.\,\eqref{aligned_zeta}-\eqref{aligned_xi}, the time derivatives of $\overline{\underline{\nu}}$, $\overline{\underline{\bm{\xi}}}$, and $\mathcal{S}$ are given by
\begin{align}
\dot{\overline{\underline{\nu}}} =& \frac{1}{\epsilon} \left(\fint\Omega_{-1}(\widetilde{\bm{\xi}}(\theta^\prime))\,d\theta^\prime - \omega_c\left((\widetilde{\Phi}_\epsilon)^{-1}(\widetilde{\underline{\zeta}}(t),\widetilde{\underline{\bm{\xi}}}(t))\right)\right) + O(1)\\
\dot{\overline{\underline{\bm{\xi}}}} =& \fint \overline{U}_\epsilon(\widetilde{\bm{\xi}}(\theta^\prime))\,d\theta^\prime \\
\dot{\mathcal{S}} = & \omega_c\left((\widetilde{\Phi}^{(N)}_\epsilon)^{-1}(\widetilde{\underline{\zeta}}(t),\widetilde{\underline{\bm{\xi}}}(t))\right)
\end{align}
The quantity $\dot{\overline{\underline{\bm{\xi}}}} = O(1)$ because $\overline{U}_\epsilon = O(1)$. The quantity $\dot{\overline{\underline{\nu}}} = O(1)$ because $\Phi_0 = \text{id}$ and $\omega_c$ is given by Eq.\,\eqref{omegac_def}. The quantity $\dot{\mathcal{S}}$ is obviously $O(1)$. Therefore $x = (\overline{\underline{\nu}},\overline{\underline{\bm{\xi}}},\mathcal{S})$ is a reasonable candidate for the slow variable.

The leading-order contributions to the time derivatives of $\widehat{\underline{\nu}}$ and $\widehat{\bm{\varrho}}$ are given by
\begin{align}
\partial_t\widehat{\underline{\nu}}(\theta,t) = &- \frac{1}{\epsilon} \,\Omega_{-1}(\overline{\underline{\bm{\xi}}}(t))\,\partial_\theta\widehat{\underline{\nu}}(\theta,t) + O(1)\\
\partial_t\widehat{\bm{\varrho}}(\theta,t) = & - \frac{1}{\epsilon}\,\Omega_{-1}(\overline{\underline{\bm{\xi}}}(t))\partial_\theta\widehat{\bm{\varrho}}(\theta,t) + O(1).
\end{align}
Therefore, if $y = (\widehat{\underline{\nu}},\widehat{\bm{\varrho}})$, $\epsilon \dot{y}= f_0(x,y) +O(\epsilon)$, with $f_0(x,y)$ given by
\begin{align}
f_0(x,y) = & (\dot{\widehat{\underline{\nu}}}_0,\dot{\widehat{\bm{\varrho}}}_0)\\
\dot{\widehat{\underline{\nu}}}_0 = & -\Omega_{-1}(\overline{\underline{\bm{\xi}}})\,\partial_\theta\widehat{\underline{\nu}}\\
\dot{\widehat{\bm{\varrho}}}_0 = & - \Omega_{-1}(\overline{\underline{\bm{\xi}}})\,\partial_\theta\widehat{\bm{\varrho}}.
\end{align}
It follows that the derivative $D_yf_0(x,y)$ is a non-vanishing multiple of the identity, and that $(x,y)$ comprise a fast-slow split for loop space dynamics. 
\end{proof}

In fact, the fast-slow split has become so simple that the coefficients defining the formal slow manifold, as well as the infinitesimal generator for the slow dynamics, may be computed explicitly to all orders in $\epsilon$.

\begin{theorem}\label{formal_slow_man}
The formal slow manifold associated with the fast-slow split given in Proposition \ref{formal_fast_slow} is given by $y_\epsilon^* = 0$. The infinitesimal generator on the formal slow manifold is given by
\begin{align}
\dot{\overline{\underline{\nu}}} =& \overline{\Omega}_\epsilon(\overline{\underline{\bm{\xi}}}) - \delta\omega_{\epsilon}(\overline{\underline{\nu}},\overline{\underline{\xi}})\\
\dot{\overline{\underline{\bm{\xi}}}} = &\overline{U}_\epsilon(\overline{\underline{\bm{\xi}}})\\
\dot{\mathcal{S}} = & \omega_c\left((\widetilde{\Phi}_\epsilon)^{-1}(\widetilde{\underline{\zeta}}^*(t),\widetilde{\underline{\bm{\xi}}}^*(t))\right)
\end{align}
where
\begin{align}
\delta\omega_{\epsilon}(\overline{\underline{\nu}},\overline{\underline{\xi}}) = \frac{1}{\epsilon}\Omega_{-1}(\overline{\underline{\bm{\xi}}}) - \frac{1}{\epsilon} \omega_c\left((\widetilde{\Phi}_\epsilon)^{-1}(\widetilde{\underline{\zeta}}^*(t),\widetilde{\underline{\bm{\xi}}}^*(t))\right),
\end{align}
and
\begin{align}
\widetilde{\underline{\zeta}}^*(\theta) =& \theta + \overline{\underline{\nu}}\\
\widetilde{\underline{\bm{\xi}}}^*(\theta) =& \overline{\underline{\bm{\xi}}}.
\end{align}
Note that $\delta\omega_{\epsilon} = O(1)$ as $\epsilon\rightarrow 0$ because $\Phi_0 = \text{id}$.
\end{theorem}
\begin{proof}
Because the formal slow manifold is unique, it suffices to check that $y_\epsilon^* = 0$ is a solution of the invariance equation. To that end, note that wherever $y = 0$, $\widetilde{\underline{\bm{\xi}}}(\theta) = \overline{\underline{\bm{\xi}}}$. The right-hand-side of the invariance equation $\epsilon Dy_\epsilon^*(x)[g_\epsilon(x,y_\epsilon^*(x))] = f_\epsilon(x,y_\epsilon^*(x))$ therefore vanishes when $y_\epsilon^* = 0$ because both $\Omega_{-1}(\overline{\underline{\bm{\xi}}}) + \epsilon \overline{\Omega}_\epsilon(\overline{\underline{\bm{\xi}}})$ and $\overline{U}_\epsilon(\overline{\underline{\bm{\xi}}})$ are independent of $\theta$. Being linear in $y_\epsilon^*$, the left-hand-side of the invariance equation also vanishes. The asymptotic series $y_\epsilon^* = 0$ is therefore the unique formal slow manifold.

\end{proof}

Theorem \ref{formal_slow_man} establishes the equivalence between dynamics on the formal slow manifold in loop space and all-orders guiding center dynamics, Eqs.\,\eqref{all_orders_gc_one}-\eqref{all_orders_gc_two}, upon making the simple identifications $\underline{\bm{\xi}} = \overline{\underline{\bm{\xi}}}$, $\underline{\zeta} = \overline{\underline{\nu}} - \mathcal{S}/\epsilon$, and then noting that the evolution of $\mathcal{S}$ decouples from the evolution of $(\overline{\underline{\nu}},\overline{\underline{\bm{\xi}}})$ on the formal slow manifold. An immediate corollary of this observation is that $\overline{\underline{\nu}}$ represents the so-called \emph{adiabatic phase}\cite{Littlejohn_1984,Littlejohn_1988,Burby_2012} associated with gyromotion.

\section{Hamiltonian structure on the formal slow manifold\label{ham_struc}}
A superconducting ring conserves the magnetic flux threading the ring's center. Charged particles in strong magnetic fields approximately exhibit the same property, although in that context the phenomenon is usually referred to as adiabatic invariance instead of superconductivity. There is good reason for this change in nomenclature; regardless of the strength of the magnetic field, a charged particle is emphatically \emph{not} a superconducting ring. Nevertheless, the proximity of the two concepts, flux conservation on the one hand and adiabatic invariance on the other, begs the following question. Since charged particles are not flux-conserving superconductors, what is the physical explanation for the behavioral similarity between the two sorts of objects?

The answer to the question is symmetry. Charged particle dynamics exhibit an approximate symmetry that, according to Noether's theorem, implies the presence of a conserved quantity that happens to be numerically equal to the magnetic flux at leading-order in perturbation theory. This explanation is of course very old and well-known. (See, for instance, Section E.5 of Ref.\,\onlinecite{Kruskal_1962}) However, the usual way of exhibiting this symmetry is rather technical and laborious, and therefore not as illuminating as one might hope from the physical point of view. In this final technical Section, I would like to elucidate the symmetry underlying a particle's approximate superconductivity, sometimes referred to suggestively as ``gyrosymmetry,"\cite{Qin_2007} in a manner that makes the symmetry itself appear almost obvious. Naturally, I aim to do this using the loop space picture of guiding center dynamics that has been the subject of this article. 

The starting point for this demonstrating is the action functional \eqref{ll_action} for loop space dynamics associated with the Lorentz force. I aim to show that obvious symmetries of this functional ultimately give rise to the symmetry of particle dynamics associated with adiabatic invariance. To that end, there are two obvious symmetries worth discussing. (1) Because the Lebesgue measure on the circle $d\theta$ is translation invariant, the value of the action does not change when the loop $(\widetilde{\bm{x}},\widetilde{\bm{v}})$ is subject to the phase shift 
\begin{align}
\widetilde{\bm{x}}&\mapsto \widetilde{\bm{x}}^{\psi_1}\label{G1x}\\
\widetilde{\bm{v}}&\mapsto \widetilde{\bm{v}}^{\psi_1},\label{G1y}
\end{align}
where $\psi_1\in S^1$ is any angle. (2) Because the phase function $S$ only appears in the action \emph{via} its time derivative, the value of the action is also unchanged when $S$ is translated according to
\begin{align}
S\mapsto S+\psi_2,\label{G2S}
\end{align}
where $\psi_2\in S^1$ is another arbitrary angle. This pair of obvious symmetries may be conveniently encoded as a single $T^2 \equiv S^1\times S^1$-action on the phase space for loop dynamics $\ell P\times S^1$,
\begin{align}
\Psi_{\psi_1,\psi_2}(\widetilde{\bm{x}},\widetilde{\bm{v}},S) = (\widetilde{\bm{x}}^{\psi_1},\widetilde{\bm{v}}^{\psi_1},S+\psi_2).
\end{align}
I will now show that $\Psi$ restricted to the subgroup $\psi_2 = 0$ is precisely the symmetry responsible for a charged particle's adiabatic invariance.

First it is convenient to leave behind the action functional \eqref{ll_action} in favor of the $1$-form $\Xi$ on $\ell P\times S^1$ given by
\begin{align}
\Xi(\widetilde{\bm{x}},\widetilde{\bm{v}},S)[\dot{\widetilde{\bm{x}}},\dot{\widetilde{\bm{v}}},\dot{S}] =&\fint \left(\frac{1}{\epsilon}\bm{A}(\widetilde{\bm{x}}(\theta))+\widetilde{\bm{v}}(\theta)\right)\cdot \dot{\widetilde{\bm{x}}}(\theta) \,d\theta\nonumber\\
&+ \dot{S}\fint \left(\frac{1}{\epsilon}\bm{A}(\widetilde{\bm{x}}(\theta)+\widetilde{\bm{v}}(\theta)\right)\cdot\partial_\theta\widetilde{\bm{x}}(\theta)\,d\theta.
\end{align}
This $1$-form is clearly related to the action functional \eqref{ll_action}, for
\begin{align}
\widetilde{A}(\widetilde{\bm{x}},\widetilde{\bm{v}},S) = \int_{t_1}^{t_2} \Xi(\widetilde{\bm{x}},\widetilde{\bm{v}},S)[\dot{\widetilde{\bm{x}}},\dot{\widetilde{\bm{v}}},\dot{S}] \,dt - \int_{t_1}^{t_2} \mathcal{H}(\widetilde{\bm{x}},\widetilde{\bm{v}})\,dt,
\end{align}
where $\mathcal{H}$ is the loop energy defined in Eq.\,\eqref{loop_energy_def}. In fact, if $\mathcal{X} = (\dot{\widetilde{\bm{x}}},\dot{\widetilde{\bm{v}}},\dot{S})$ denotes the infinitesimal generator for loop space dynamics, c.f. Eqs.\,\eqref{lorentz_loop_v}-\eqref{lorentz_loop_S}, the Euler-Lagrange equations associated with $\widetilde{A}$ may be written
\begin{align}
\iota_{\mathcal{X}}\mathbf{d}\Xi  = - \mathbf{d}\mathcal{H},\label{presymp_ham}
\end{align}
where $\mathbf{d}$ denotes the exterior derivative on $\ell P\times S^1$. Equation \eqref{presymp_ham} is an example of a \emph{presymplectic Hamilton's equation}. Therefore $\Xi$, together with $\mathcal{H}$, geometrically encode the information contained in the functional $\widetilde{A}$.

The invariance of $\widetilde{A}$ under the $T^2$-action $\Psi$ is equivalent to the pair of pullback relations
\begin{align}
\Psi_{\psi_1,\psi_2}^*\Xi =& \Xi\label{xi_invariance}\\
\Psi_{\psi_1,\psi_2}^*\mathcal{H} = & \mathcal{H}.\label{H_invariance}
\end{align}
Therefore, if we define the infinitesimal generators 
\begin{align}
\bm{\partial}_1(\widetilde{\bm{x}},\widetilde{\bm{v}},S) =& \frac{d}{d\lambda}\bigg|_0 \Psi_{\lambda,0}(\widetilde{\bm{x}},\widetilde{\bm{v}},S)\\
\bm{\partial}_2(\widetilde{\bm{x}},\widetilde{\bm{v}},S) = &\frac{d}{d\lambda}\bigg|_0\Psi_{0,\lambda}(\widetilde{\bm{x}},\widetilde{\bm{v}},S),
\end{align}
Noether's theorem implies the functionals 
\begin{align}
\mathcal{I}_1 = &\iota_{\bm{\partial}_1}\Xi\\
\mathcal{I}_2 = & \iota_{\bm{\partial}_2}\Xi
\end{align}
are each constant along trajectories of loop space dynamics. Curiously, these functionals are \emph{each} equal to the (normalized) loop action \eqref{loop_action_def},
\begin{align}
\mathcal{I}_1 = \mathcal{I}_2 = \fint \left(\frac{1}{\epsilon}\bm{A}(\widetilde{\bm{x}}(\theta)+\widetilde{\bm{v}}(\theta)\right)\cdot\partial_\theta\widetilde{\bm{x}}(\theta)\,d\theta.
\end{align}


Now consider the formal slow manifold $\Gamma \subset \ell P\times S^1$. Let $\mathcal{F}_t$ denote the loop space dynamics flow. Being a formally invariant set, the flow on loop space maps $\Gamma$ into itself, i.e. $\mathcal{F}_t(\Gamma) = \Gamma$. Therefore the set $\Psi_{\psi_1,\psi_2}(\Gamma)\equiv \Gamma_{\psi_1,\psi_2}$ satisfies
\begin{align}
\mathcal{F}_t(\Gamma_{\psi_1,\psi_2}) = \mathcal{F}_t(\Psi_{\psi_1,\psi_2}(\Gamma)) = \Psi_{\psi_1,\psi_2}(\mathcal{F}_t(\Gamma)) = \Gamma_{\psi_1,\psi_2},
\end{align}
where I have used the commutativity of the flow $\mathcal{F}_t$ and the $T^2$-action $\Psi_{\psi_1,\psi_2}$ implied by Eqs.\,\eqref{xi_invariance}-\eqref{H_invariance}. In other words $\Gamma_{\psi_1,\psi_2}$ is an invariant set for each $(\psi_1,\psi_2)\in T^2$. In the fast-slow coordinates $(x,y)$ on $\ell P\times S^1$, and for sufficiently small $(\psi_1,\psi_2)$, $\Gamma_{\psi_1,\psi_2}$ is therefore a formally invariant set given as the graph of some function $Y_{\psi_1,\psi_2}(x)$.  Moreover, the fact that $\Psi_{\psi_1,\psi_2}$ does not depend on $\epsilon$ implies that $Y_{\psi_1,\psi_2}$ must be a formal power series in $\epsilon$. By the uniqueness of the formal slow manifold, this means that $\Gamma_{\psi_1,\psi_2} = \Psi_{\psi_1,\psi_2}(\Gamma) = \Gamma$ is equal to the formal slow manifold, i.e. that $\Gamma$ is $T^2$-invariant to all orders in $\epsilon$. (Compare this argument with Kruskal's ``Theorem of phase independence" in Section C.1 of Ref.\,\onlinecite{Kruskal_1962}.)

Let $\gamma:\Gamma\rightarrow \ell P\times S^1$ be the inclusion map. Because $\Gamma$ is a formally invariant set, the presymplectic Hamilton's equation \eqref{presymp_ham} implies
\begin{align}
&\gamma^*(\iota_{\mathcal{X}}\mathbf{d}\Xi) =- \gamma^*\mathbf{d}\mathcal{H}\nonumber\\
\Rightarrow & \iota_{\mathcal{X}_\Gamma}\mathbf{d}\Xi_\Gamma =- \mathbf{d}\mathcal{H}_\Gamma,
\end{align}
where $\mathcal{X}_{\Gamma}$ is the infinitesimal generator $\mathcal{X}$ restricted to the formal slow manifold $\Gamma$, and $\Xi_\Gamma,\mathcal{H}_\Gamma$ are the pullbacks of the $1$-form $\Xi$ and functional $\mathcal{H}$ along $\gamma$. Because $\Gamma$ is $T^2$-invariant, the pullback relations \eqref{xi_invariance}-\eqref{H_invariance} imply analogous pullback relations on $\Gamma$:
\begin{align}
\Psi_{\psi_1,\psi_2}^*\Xi_{\Gamma } &= \Xi_\Gamma\\
\Psi_{\psi_1,\psi_2}^*\mathcal{H}_\Gamma & = \mathcal{H}_\Gamma.
\end{align}
Noether's theorem applied to $\Psi_{\psi_1,\psi_2}$ \emph{restricted to the formal slow manifold} therefore implies that 
\begin{align}
\mathcal{J}_1 &= \iota_{\bm{\partial}_1}\Xi_\Gamma\\
\mathcal{J}_2 & = \iota_{\bm{\partial}_2}\Xi_{\Gamma}.
\end{align}
Are each constant along trajectories contained in the formal slow manifold. Because $\iota_{\bm{\partial}_k}\Xi_\Gamma = \gamma^*(\iota_{\bm{\partial}_k}\Xi) = \gamma^*\mathcal{I}_k$, $\mathcal{J}_1$ and $\mathcal{J}_2$ are each equal to the normalized loop action restricted to the formal slow manifold. Because dynamics on the formal slow manifold is the same thing as guiding center dynamics, the obvious symmetry $\Psi_{\psi_1,\psi_2}$ on the phase space for loop dynamics is now shown to be responsible for a nontrivial conservation law for guiding center dynamics. My argument will therefore be complete if I can show that $\mathcal{J}_1 = \mathcal{J}_2 = \mathcal{J}$ is equal to the magnetic flux at leading-order in $\epsilon$.

To that end, I will demonstrate even more by explicitly recovering the Hamiltonian formulation of guiding center dynamics due to Littlejohn from the restricted $1$-form $\Xi_\Gamma$ and the restricted Hamiltonian $\mathcal{H}_\Gamma$. First note that the equality of the two Noether invariants $\mathcal{J}_1,\mathcal{J}_2$ has the remarkable consequence that the difference of the infinitesimal generators $\bm{\Delta} = \bm{\partial_1} - \bm{\partial_2}$ lies in the kernel of the closed $2$-form $\mathbf{d}\Xi_\Gamma$. Indeed,
\begin{align}
\iota_{\bm{\Delta}}\mathbf{d}\Xi_\Gamma &= \iota_{\bm{\partial}_1}\mathbf{d}\Xi_\Gamma - \iota_{\bm{\partial}_2}\mathbf{d}\Xi_\Gamma \nonumber\\
& = \mathcal{L}_{\bm{\partial}_1}\Xi_\Gamma - \mathcal{L}_{\bm{\partial}_2}\Xi_\Gamma - \mathbf{d}\mathcal{J}_1 +\mathbf{d}\mathcal{J}_2\nonumber\\
& = 0.
\end{align}
(Note that the identity $\mathcal{L}_{\bm{\partial}_k}\Xi_\Gamma = 0$ follows from differentiating Eq.\,\eqref{xi_invariance}.) Therefore $\mathbf{d}\Xi_\Gamma$ is not a symplectic form. This suggests that in order to recover Littlejohn's symplectic formulation of guiding center dynamics, it is necessary to first quotient $\ell P \times S^1$ by the foliation tangent to the kernel of $\mathbf{d}\Xi_\Gamma$. 

Because it is not immediately clear whether the dimension of $\mathbf{d}\Xi_\Gamma$'s characteristic foliation is greater than $1$, it is helpful to first quotient by the subfoliation tangent to $\bm{\Delta}$. By a slight abuse of notation, I will denote the latter foliation by $\bm{\Delta}$.
If the descendent of $\mathbf{d}\Xi_\Gamma$ on $\Gamma/\bm{\Delta}$ is non-degenerate, this would imply that $\bm{\Delta}$ frames the kernel of $\mathbf{d}\Xi_\Gamma$ and that the quotient by $\bm{\Delta}$ produces a symplectic space. Otherwise, a further quotient may be necessary. In the case at hand it is reasonable to suspect that only the quotient by $\bm{\Delta}$ is necessary; because the dimension of the formal slow manifold $\Gamma$ is $7$ and the dimension of the foliation tangent to $\bm{\Delta}$ is $1$, the quotient by $\bm{\Delta}$ will be $6$-dimensional, which is the
same dimension as Littlejohn's symplectic phase space.

In order to explicitly carry out the quotient by $\bm{\Delta}$, it is necessary to have an explicit expression for $\Psi_{\psi_1,\psi_2}$ restricted to the formal slow manifold, and especially useful to have this expression in the natural coordinates $x$ on $\Gamma$. To find this expression, observe that because $\Gamma$ is $T^2$-invariant, there must be a mapping $\varphi_{\psi_1,\psi_2}:X\rightarrow X$ such that 
\begin{align}
\Psi_{\psi_1,\psi_2}(x,y_\epsilon^*(x)) = (\varphi_{\psi_1,\psi_2}(x),y_\epsilon^*(\varphi_{\psi_1,\psi_2}(x)))
\end{align}
for all $x\in X$. The mapping $\varphi_{\psi_1,\psi_2}:X\rightarrow X$ is precisely the $T^2$-action restricted to $\Gamma$ expressed in the coordinates $x$. Also observe that after applying $\Psi_{\psi_1,\psi_2}$ to an arbitrary point $(x,y)$, the slow variable $x = (\overline{\bm{x}},\overline{u},w_1,w_2,\mathcal{S})$ transforms according to
\begin{align}
\overline{\bm{x}}&\mapsto \overline{\bm{x}}\\
\overline{u}&\mapsto \overline{u}\\
w_1&\mapsto w_1\cos\psi_1 + w_2\sin\psi_1\\
w_2&\mapsto w_2\cos\psi_1 - w_1\sin\psi_1\\
\mathcal{S}&\mapsto \mathcal{S} + \epsilon \psi_2.
\end{align}
(The transformation rules for $w_1,w_2$ may be summarized in vector notation as $\bm{w}_\perp\mapsto \bm{w}_\perp\cos\psi_1 + \bm{w}_\perp\times\bm{b}\,\sin\psi_1$.) In particular, the transformation of the slow variable $x$ does not depend on the fast variable $y$. It follows that the $T^2$-action on $x$-space is given by
\begin{align}
\varphi_{\psi_1,\psi_2}(\overline{\bm{x}},\overline{u},w_1,w_2,\mathcal{S}) = (\overline{\bm{x}},\overline{u},w_1\cos\psi_1 + w_2\sin\psi_1,w_2\cos\psi_1 - w_1\sin\psi_1,\mathcal{S} + \epsilon \psi_2),
\end{align}
whence the quotient by $\bm{\Delta}$, $\pi:\Gamma\rightarrow \Gamma/\bm{\Delta}$, may be identified as
\begin{align}
\pi(\overline{\bm{x}},\overline{u},w_1,w_2,\mathcal{S}) = (\overline{\bm{x}},\overline{u},u_1,u_2),
\end{align}
where
\begin{align}
\bm{u}_\perp = u_1\bm{e}_1(\overline{\bm{x}}) + u_2\bm{e}_2(\overline{\bm{x}}) = \cos(\mathcal{S}/\epsilon)\bm{w}_\perp + \sin(\mathcal{S}/\epsilon)\bm{w}_\perp\times\bm{b}.
\end{align}
In particular, a useful section of $\pi$ is given by $s: \Gamma/\bm{\Delta}\rightarrow \Gamma$, 
\begin{align}
s(\overline{\bm{x}},\overline{u},u_1,u_2) = (\overline{\bm{x}},\overline{u},u_1,u_2,0).
\end{align}

Because $\iota_{\bm{\Delta}}\mathbf{d}\Xi_\Gamma = 0$ and $\Psi_{\psi_1,\psi_2}^*\Xi_\Gamma = \Xi_\Gamma$, the $2$-form $\mathbf{d}\Xi_\Gamma$ on $\Gamma$ descends to the $2$-form $\mathbf{d}\Xi_{\Gamma/\bm{\Delta}}$ on $\Gamma/\bm{\Delta}$, where the $1$-form $\Xi_{\Gamma/\bm{\Delta}} = s^*\Xi_\Gamma = s^*\gamma^*\Xi = (\gamma\circ s)^*\Xi$. An explicit expression for $\Xi_{\Gamma/\bm{\Delta}}$ modulo an exact $1$-form is
\begin{align}
&\Xi_{\Gamma/\bm{\Delta}}(\overline{\bm{x}},\overline{u},u_1,u_2)[\dot{\overline{\bm{x}}},\dot{\overline{u}},\dot{u}_1,\dot{u}_2]\nonumber\\ = & \bigg(\frac{1}{\epsilon}\bm{A}(\overline{\bm{x}}) + \overline{u}\bm{b}(\overline{\bm{x}}) + \overline{\bm{v}}_{\perp\epsilon}^*+ \int_0^1 \fint \bm{B}(\overline{\bm{x}}+ \lambda\epsilon \widehat{\bm{\rho}}^*_\epsilon)\times \widehat{\bm{\rho}}^*_\epsilon\,d\theta\,d\lambda\bigg)\cdot \dot{\overline{\bm{x}}}\nonumber\\
& + \epsilon\fint \bigg( \widehat{\bm{v}}^*_\epsilon + \int_0^1 \bm{B}(\overline{\bm{x}}+\lambda\epsilon\widehat{\bm{\rho}}_\epsilon^*)\times\widehat{\bm{\rho}}_\epsilon^*\,\lambda d\lambda\bigg)\cdot D\widehat{\bm{\rho}}_\epsilon^*[\dot{\overline{\bm{x}}},\dot{\overline{u}},\dot{u}_1,\dot{u}_2]\,d\theta\\
& = \bigg(\frac{1}{\epsilon}\bm{A}(\overline{\bm{x}})+\overline{u}\bm{b}(\overline{\bm{x}}) +\epsilon \bm{W}\bigg)\cdot\dot{\overline{\bm{x}}} + \epsilon\,\mu_0 \frac{u_2 \,du_1 - u_1 \,du_2}{u_1^2 + u_2^2} + O(\epsilon^2),\label{restricted_2form}
\end{align}
where $\mu_0 = \left(\frac{|\bm{u}_\perp|^2}{2|\bm{B}(\overline{\bm{x}})|}\right)$ and
\begin{align}
\bm{W} =& -\frac{(\mu_0\nabla|\bm{B}| + \overline{u}^2\bm{\kappa})\times\bm{b}}{|\bm{B}|} + \frac{1}{2} \fint (\widehat{\bm{\rho}}_0^*\cdot\nabla\bm{B})\times \widehat{\bm{\rho}}_0^*\,d\theta + \frac{1}{2}\fint (\nabla\widehat{\bm{\rho}}_0^*)\cdot\widehat{\bm{v}}_0^*\,d\theta\nonumber\\
& = -\frac{3}{2}\frac{\mu_0\nabla |\bm{B}|\times\bm{b}}{|\bm{B}|} - \frac{\overline{u}^2\bm{\kappa}\times\bm{b}}{|\bm{B}|} - \frac{1}{2}\mu_0\tau\,\bm{b} - \mu_0 \bm{R}.
\end{align}
Upon taking an exterior derivative, Eq.\,\eqref{restricted_2form} reproduces the symplectic form for guiding center theory derived by Littlejohn modulo terms of $O(\epsilon)$ in $\bm{W}$. This is not a contradiction. Littlejohn's derivation made use of near-identity coordinate transformations that are not uniquely determined, i.e. the transformations depended on a number of arbitrary parameters. Different choices for those parameters would be necessary to recover the result \eqref{restricted_2form} from the near-identity coordinate transformation approach.

The residual part of the symmetry $\Psi$ that survives when passing to the quotient is the transformation
\begin{align}
u_1\mapsto & u_1\cos\psi + u_2\sin\psi\\
u_2\mapsto & u_2\cos\psi - u_1\sin\psi.
\end{align}
According to Eq.\eqref{restricted_2form}, the conserved quantity associated with this residual symmetry is given by
\begin{align}
\mathcal{J}_1 = \epsilon \mu_0 + O(\epsilon^2).
\end{align}
Because $\mu_0 = |\bm{B}||\widehat{\bm{\rho}}_0^*|^2/2$, $\mathcal{J}_1$ is proportional to the magnetic flux passing through the loop $\overline{\bm{x}}+\epsilon \widehat{\bm{\rho}}_0^*(\theta)$, as claimed. The conclusion is that the adiabatic invariant for charged particle dynamics in a strong magnetic field is the Noether conserved quantity associated with the symmetry of loop space dynamics under phase shift $\widetilde{z}\mapsto \widetilde{z}^\psi$.

\section{Discussion}

The loop space picture of guiding center dynamics developed in this article is closely related to the nonlinear WKB expansion of Kruskal\cite{kruskal58} and the two-timescale technique described by Hazeltine and Waelbroeck in Ref.\,\onlinecite{Hazeltine_Waelbroeck_2004}. Kruskal introduced the ansatz
\begin{align}
\bm{x}(t) = \sum_{k=-\infty}^\infty \epsilon^{|k|}\bm{X}_k(t)\,\exp(i k C(t)/\epsilon)\label{ansatz_kruskal}
\end{align}
for the spatial location of a charged particle, where $\bm{X}_k$ and $C$ were allowed to be formal power series in $\epsilon$. Hazeltine and Waelbroeck introduced the ansatz
\begin{align}
\bm{x}(t) =& \bm{X}(t) + \epsilon \bm{\rho}(\bm{X}(t),\bm{U}(t),t,\gamma(t))\label{ansatz_H1}\\
\bm{v}(t) = & \bm{U}(t) + \bm{u}(\bm{X}(t),\bm{U}(t),t,\gamma(t)),\label{ansatz_H2}
\end{align}
for the spatial location and velocity of a charged particle, where $\bm{U},\bm{\rho},\bm{u},\gamma$ were allowed to be formal power series in $\epsilon$, and the profile functions $\bm{\rho},\bm{u}$ were assumed periodic in $\gamma$ with zero average. Of course, being periodic, $\bm{\rho},\bm{u}$ may also be written
\begin{align}
\bm{\rho} = \sum_{k\neq 0} \bm{\rho}_k(\bm{X}(t),\bm{U}(t),t)\,\exp(i k\gamma)\\
\bm{u} = \sum_{k\neq 0} \bm{u}_k(\bm{X}(t),\bm{U}(t),t) \,\exp(ik\gamma),
\end{align}
which establishes a close link to Kruskal's ansatz \eqref{ansatz_kruskal}. Apparently each of these representations of the solution to Newton's equation $\epsilon \ddot{\bm{x}} = \dot{\bm{x}}\times \bm{B}(\bm{x})$ involve evaluating a time dependent, parameterized loop (characterized either by the coefficients $\bm{X}_k$ or $\bm{\rho}_k,\bm{u}_k$) at a rapidly rotating phase (either $C$ or $\gamma$). Therefore there can be no doubt that loop space is playing a role in each of these approaches. 

On the other hand, these approaches differ from the formal slow manifold approach described in this article for the following reasons. 
\\ \\
\noindent (1) Neither approach recognizes the link between its collection of formal asymptotic series and the non-perturbative notion of loop space dynamics. In particular, neither approach establishes that guiding center dynamics is equivalent to loop space dynamics restricted to a formal slow manifold. The conceptual framework supported by the formal slow manifold picture is therefore missing from Kruskal's and Hazeltine and Waelbroeck's work. For instance, as discussed in Section \ref{ham_struc}, the genuine simplicity of the symmetry underlying adiabatic invariance only manifests itself in the context of loop space dynamics. Moreover, the prospect of using the numerical method from Ref.\,\onlinecite{Gear_2006} to capture high-order guiding center effects without resorting to the laborious machinations of perturbation theory would not emerge in absence of the formal slow manifold picture.
\\ \\
\noindent (2) While in Kruskal's approach each of the $\bm{X}_k$'s is expanded in an asymptotic series, in the formal slow manifold approach only the fast variable restricted to the formal slow manifold $y = y_\epsilon^*$ is expanded in such a manner; the slow variable $x$ is never subject to asymptotic expansion. In fact Kruskal does not identify the fast slow split given in Theorem \ref{fast_slow_proof} at all. The main drawback of this ``expand everything" approach is that it obscures the phase space geometry underlying the problem with needless additional algebraic manipulations. Indeed, Kruskal comments that additional technical work due to Gardner and Berkowitz is required to prove that his calculation produces a valid asymptotic expansion of a solution to Newton's equation. In contrast, from the perspective of fast-slow systems the error estimates required to prove such a result may be formulated in a general context using little more than Gronwall's inequality.\cite{MacKay_2004} That said, Hairer and Lubich\cite{Hairer_Lubich_2018} have recently managed to apply Kruskal's ansatz to the problem of adiabatic invariance for a particular structure-preserving numerical integrator for charged particle dynamics; the ansatz is used in a proof that the integrator preserves a modified adiabatic invariant over extremely large time intervals.
\\ \\
\noindent (3) While Hazeltine and Waelbroeck aim to parameterize the fluctuating position $\bm{\rho}$ and fluctuating velocity $\bm{u}$ using the mean position $\bm{X}$ and mean velocity $\bm{U}$, the slow manifold approach parameterizes the fast variable $y$ using the slow variable $x\neq (\bm{X},\bm{U})$. In particular, the parameter $x$ involves pieces of the first harmonic of the fluctuating velocity, and does not involve the perpendicular components of the mean velocity. A first guess at a resolution of this apparent discrepancy is that the implicit function theorem might be able reparameterize the graph $y=y_\epsilon^*(x)$ by the variables $(\overline{\bm{x}},\overline{u},\overline{v}_1,\overline{v}_2)$, which are equivalent to Hazeltine and Waelbroeck's $(\bm{X},\bm{U})$. However, expressions \eqref{v1bar_star1}-\eqref{v2bar_star1} show that such an inversion is very poorly conditioned in general, and impossible in a uniform magnetic field. It therefore seems $(\bm{X},\bm{U})$ is a problematic choice for parameterizing the fluctuating position and velocity. That this is true can also be seen in the details of the guiding center calculation presented in Chapter 2 of Ref.\,\onlinecite{Hazeltine_Waelbroeck_2004}. While the goal of the calculation was to determine the dependence of $\bm{\rho}$ and $\bm{u}$ on $(\bm{X},\bm{U})$, an unexpected constraint on the perpendicular components of $\bm{U}$ appears in Eq.\, (2.30). In addition, as a consequence of mischaracterizing the general solution of Eq.\,(2.28), the derivation fails to recognize that there are two undetermined parameters (instead of one) in the leading-order contribution to $\bm{u}$. If the roles of the constrained components of $\bm{U}$ and the unconstrained components of $\bm{u}$ were merely exchanged, the method of Hazeltine and Waelbroeck would reproduce the steps in computing the formal slow manifold in loop space.  In other words, while the goal of the calculation in Ref.\,\onlinecite{Hazeltine_Waelbroeck_2004} seems to be flawed, the calculation itself seems to be suggesting that identifying the fast slow split as in Theorem \ref{fast_slow_proof} is the way to fix the problem!

It is also interesting to compare the approach used in Section \ref{ham_struc} to identify the noncanonical Hamiltonian structure of guiding center dynamics with Littlejohn's approach in Refs.\,\onlinecite{Littlejohn_1981,Littlejohn_1982}. Because Littlejohn worked in particle space rather than loop space, his strategy for identifying the Hamiltonian structure was to exploit the coordinate covariance of the symplectic Hamilton's equation, i.e. that the equation $\iota_\mathcal{X}\mathbf{d}\theta = -\mathbf{d}H$ has the same form in any coordinate system. In contrast, the strategy used in Section \ref{ham_struc} to find the Hamiltonian structure made use of the form-invariance of the (pre-)symplectic Hamilton's equation under restriction to invariant sets.

In a forthcoming publication, I will report on exploiting the formal slow manifold picture of guiding center dynamics for the sake of building a novel numerical scheme for efficiently simulating the slow drift of strongly magnetized charged particles. I have managed to show that applying the implicit-midpoint time integration scheme to loop space dynamics expressed in terms of the fast and slow variables $(x,y)$ leads to a nonlinearly implicit energy-conserving scheme that does not suffer from the preconditioning problem that usually plagues implicit integrators applied to stiff problems. Moreover, the scheme is provably accurate when taking timesteps much larger than the cyclotron period $\epsilon$ provided initial conditions are chosen to lie approximately on the formal slow manifold. This integrator is currently being optimized for the purpose of employing Gear \emph{et. al}'s technique for numerically selecting initial conditions on the formal slow manifold with any desired accuracy. The ultimate goal of this undertaking is to develop the first charged particle simulation tool that is able to resolve high-order guiding center effects while stepping over the cyclotron period. Such high-order effects appear to play an important role in the dynamics of so-called runaway electrons generated in magnetic traps.\cite{Liu_2018}

\section{Acknowledgements}
Research presented in this article was supported by the Los Alamos National Laboratory LDRD program under project number 20180756PRD4.


%


\bibliography{cumulative_bib_file.bib}



\end{document}